\documentclass[12pt]{amsart}%
\usepackage{graphicx}
\usepackage{amsmath}
\usepackage{amsfonts}
\usepackage{amssymb}
\usepackage{fancyvrb}
\usepackage{longtable}
\usepackage{natbib}
\newtheorem{theorem}{Theorem} 

\newtheorem{lemma}[theorem]{Lemma}

\usepackage[doublespacing]{setspace}
\providecommand{\U}[1]{\protect\rule{.1in}{.1in}}
\usepackage{color}

\begin{document}

\title[Refundable Income Annuities]{Refundable Income Annuities: \\ Feasibility of Money-Back Guarantees}

\author[Milevsky]{Moshe A. Milevsky }
\author[Salisbury]{Thomas S. Salisbury}
\thanks{Milevsky is a Professor of Finance at the Schulich School of Business, York University, and Executive Director of the IFID Centre. Salisbury is a Professor in the Department of Mathematics and Statistics at York University. The authors acknowledge funding from the IFID Centre (Milevsky) and from NSERC (Salisbury) as well as helpful comments from Narat Charupat, Branislav Nikolic, and participants at the (virtual) IME Congress in early July 2021. The corresponding author (Salisbury) can be reached at: salt@yorku.ca}

\date{Version: 26 August 2021}

\begin{abstract}[Short version:] {\em Refundable} income annuities (IA), such as {\em cash-refund} and {\em instalment-refund}, differ in material ways from the {\em life-only} version beloved by economists. In addition to lifetime income they guarantee the annuitant or beneficiary will receive their money back albeit slowly over time. We document that {\em refundable} IAs now represent the majority of sales in the U.S., yet they are mostly ignored by insurance and pension economists. And, although their pricing, duration, and money's-worth-ratio is complicated by {\em recursivity} which will be explained, we offer a path forward to make refundable IAs tractable.

A key result concerns the market price of cash-refund IAs, when the actuarial present value is grossed-up by an insurance loading. We prove that price is counterintuitively no longer a declining function of age and older buyers might pay more than younger ones. Moreover, there exists a threshold valuation rate below which no price is viable. This may also explain why inflation-adjusted IAs have all but disappeared. 
\end{abstract}

\maketitle
\thispagestyle{empty}

\newpage
\setcounter{page}{1}

\hbox{  }
\vspace{0.5in}

\begin{singlespace}
\noindent 
{\bf Title: Refundable Income Annuities: \\ Feasibility of Money-Back Guarantees}
\end{singlespace}

\vspace{0.1in}

\noindent  {\bf Authors: Moshe A. Milevsky and Thomas S. Salisbury}

\vspace{0.5in}

\noindent {\bf Abstract} {\it [Long version]: }

This paper offers some basic theorems and observations related to: (i) cash refund and (ii.) instalment refund income annuities (IAs). We distinguish between the {\em price} paid by annuitants and the {\em actuarial present value} of the future payments received by the annuitant. These differ when a {\em loading} is applied, but are obviously the same in the unloaded case. We start by proving that unloaded prices are a declining function of issue age $x$ and valuation rates $r$, neither of which are trivial statements when dealing with refundable products. We also compute measures of {\em annuity duration} and show how they differ from traditional interest rate sensitivity, again due to the refund feature.

\vspace{0.1in}

In terms of contribution, our main (novel) result relates to loaded cash-refund IAs, that is once the {\em value} is grossed-up by $(1+\pi)$ and converted to a market {\em price} paid by or charged to annuitants. Under mild assumptions the price of a cash-refund IA is no longer a declining function of age and older annuitants might have to pay more than younger ones for the same lifetime of income. More importantly, there exists a threshold valuation rate $r_{\pi,x}$, at which the pricing function of a {\em cash-refund} IA is no longer viable. Practically speaking, if-and-when the insurance company or pension fund assumed investment return (i.e. portfolio yield) is projected to fall under $r_{\pi,x}$, they are no longer able to offer {\em refundable} IAs at any price. Finally, our paper might help explain why inflation-adjusted income annuities have all-but disappeared from the retirement landscape. Real interest rates are currently far under our $r_{\pi,x}$, which means that real refunds are no longer feasible.

\vspace{0.5in}

\begin{singlespace}
\begin{flushleft}
{\bf JEL:}  G12 (Asset Pricing), G22 (Actuarial Studies), G52 (Insurance) \\
{\bf Keywords:} 	Pensions, Annuities, Retirement Planning, Bond Duration.
\end{flushleft}
\end{singlespace}

\newpage

\section{Introduction and Motivation}
\label{sec:intro}

Per the SECURE Act of 2019, also known as the {\em Setting Every Community Up for Retirement Enhancement Act}, the U.S. Department of Labor has been tasked with creating a set of technical regulations and mathematical assumptions for reporting the account balance of all Defined Contribution (DC) retirement plans in the U.S. as an income stream, together with the usual mark-to-market value of all securities held at the end of the quarter\footnote{See, for example, https://www.groom.com/resources/lifetime-income-provisions-under-the-secure-act/}. 

\vspace{0.1in}

This regulation comes into effect in late September 2021 and will affect 76.8 million Americans with an ERISA-regulated DC account. As of the effective date and assuming it goes thru as planned every single plan sponsor and administrator will be mandated to calculate how much life annuity income the balance would buy, if the participant decided to convert the lump sum into a lifetime of income. This applies to participants at all ages and not just at or near retirement. 

\vspace{0.1in}

Thus for example, a 40-year-old employee working at a meat packaging plant in Nebraska will be informed on their quarterly 401(k) statement that the \$100,000 balance in their 401(k) plan would be equivalent to (only) \$5,000 of lifetime income because the actuarial pricing factor -- the denominator by which the \$100,000 account balance is divided -- is 20.

\vspace{0.1in}

The policy objective of this entire exercise – which the Department of Labor is estimating will cost employers \$200 million to implement in just the first year -- is to nudge participants to anchor their limited attention on a standard of living instead a lump-sum of money\footnote{See https://www.govinfo.gov/app/details/FR-2020-09-18/2020-17476 for the Interim Final Rules issued by the Department of Labor and printed in the Federal Register}. This is a laudable goal, familiar to financial economists as the foundation of the life-cycle model.

\vspace{0.1in}

To be clear, the calculation by which the ongoing balance is converted into a lifetime income will be based on a {\em theoretical formula} for the value of life annuity and the math itself has been the subject of much controversy and debate over the last year. While insurance companies that hope to sell annuities to many of those 76.8 million participants welcome the spotlight on life annuities, traditional asset managers who offer their own decumulation solutions aren’t as welcoming\footnote{See for example https://www.ft.com/content/cbbd43cb-e80c-4e81-8fb2-dfced5f7b62e}.

\vspace{0.1in}

To this end, the Department of Labor has released (and then updated) a 30-page document with proposals for the discount rates, the mortality rates, the annuity riders and even the magnitude of loadings, after receiving public comments about these matters. In other words, there is an enormous amount of interest among a vast array of businesses about how to value a stream of lifetime income. Indeed, a topic that was of interest to a small number of specialized insurance actuaries has gone mainstream. Quite ironically and as most pension economists are well-aware, few if any participants in DC plans voluntarily purchase life annuities with their account balances, which is yet another interesting aspect of the proposed regulation. {\em Will it help?}

\vspace{0.1in}

Indeed, the literature has long-known and puzzled over the thin market for voluntary annuities and the fact few people actively annuitize their wealth at retirement\footnote{See \citet{M1988}}. Annuitization is the process of exchanging a lump-sum of money for a stream of lifetime income that can't be outlived, via products sold by insurance companies but also offered by traditional Defined Benefit (DB) pension funds. The formal literature began with \citet{Y1965} who proved that, absent bequest motives, a utility-maximizing consumer should annuitize all wealth. And ever since \citet{Y1965}, many explanations have been proposed for why most people don't actually do this, and just as many remedies have been offered to fix the annuity market, such as mandating income projections on DC account statements\footnote{See \citet{BPT2011}.}. 

\vspace{0.1in}

We should note that the academic annuity literature is mostly about income annuities (IAs) which are different from tax-sheltered accumulation annuities intended as long-term savings vehicles. And, although the original paper \citet{Y1965} -- and thousands of models and papers that have followed in its path -- deal with actuarial notes (or instantaneous tontines), most economist treat IAs as the nearest real-world substitute, and the above-noted proposals are about IAs.

\vspace{0.10in}

What is less known among economists -- and the main impetus for our paper -- is that of the \$8-\$10 billion of annual premium flowing into IAs in the U.S. in recent years, the majority is allocated to products with a {\em money-back guarantee}, which of course reduces the embedded longevity insurance, risk pooling and mortality credits. According to data from CANNEX Financial Exchanges\footnote{CANNEX Financial Exchanges is a data vendor that provides annuity quotes and conducts a quarterly survey of the types of annuities sold. Note that both authors have a financial relationship with CANNEX.}, the most popular type of IA quoted\footnote{Actual sales are expected to be even more skewed than quotes, since duplicate quotes are more likely to be purchasers of refundable IAs assessing the guarantee cost. See \citet{CKM2016} for more details concerning this dataset.} is one in which the annuitant {\bf or} his/her beneficiary is guaranteed their money back. 

As noted in Table \ref{salesdata} of this paper\footnote{All tables and figures are placed at the end of the paper.}, in general there are three structurally distinct types of IAs. The first class is a {\em life only} product under which death terminates the income flow. Again, those are closest to the actuarial notes of \citet{Y1965}. The second class is an IA in which payments are guaranteed for a {\em fixed number} of years, to a beneficiary if the annuitant is no longer alive. The third is {\em refundable} IAs in which the annuitant has a guarantee that their entire premium will be returned, either in a lump-sum of cash or instalments. Notice the trend in Table \ref{salesdata} over time during the last 10 years. In 2011 life only (a.k.a. Yaari) IAs were more popular than refundable IAs. But by early 2021 the majority of price quotes are in cash or instalment refund products.

\begin{center}
\fbox{{\bf Table \ref{salesdata} Placed Here}}
\end{center}

Note that our focus here -- and the data displayed in the table -- is on single-life annuities. But joint-life annuities, which we do not treat in this work, also display the same pattern over time, towards cash-refund and instalment refund features. Naturally the demand and need for a {\em money back} guarantee is reduced because it's only after the second death that the annuity payment is extinguished. Either way, we leave the analysis of two (or more) lives to other authors and papers.

\vspace{0.1in}

To be clear, here is a simple example that should explain the mechanics of a cash versus instalment IA (without loading) and by extension the recursivity problem which is the focus of this paper. Assume an IA price that is: $a=12.5$ per \$1 of lifetime income. This means that if the annuitant deposits or provides a premium of $P=\$100,000$, it leads to an income of $\$C=P/a=$ 8,000 per year for life. Now, if this particular IA is a life only product, then death extinguishes the contract and the beneficiaries aren't entitled to any residual value. But, if this is a cash-refund IA, and if the annuitant happens to die after exactly 10 years of payments the estate or beneficiary receives a lump-sum cash payout of: $(\$100,\!000-\$80,\!000) = \$20,\!000$. But, if death occurs after year 12.5, the estate receives nothing. 

\vspace{0.1in}

So, with a cash-refund IA, the premium $P$ generates a lifetime cash-flow $C$, plus a death benefit $(P-TC)_+$, where $T$ is the stochastic date of death. We now arrive at the main problem and motivation for this paper. Note that the $P$ appears on both sides of the valuation equation, hence the recursivity. Likewise, if the IA is of the installment refund type, then if at the time of death $P<TC$, the annual payments of $C$ continue until time $T=P/C$, by which time the entire premium is returned. After time $T=P/C$, which in the above case would be $T=12.5$ years, all payments cease. One again, this induces a recursivity problem, all of which will be addressed.

\vspace{0.10in}

One additional motivation for this paper (in mid 2021) is that the consumer trend towards refundable income annuities, is on a collision course with low (and possibly negative) interest rates. At the time of writing, the yields on long-term European government bonds are (mostly) negative\footnote{Source: European Central Bank \& IMF Report, December 2020. For example, in late June 2021 the 10-year rate on Greek government bonds fell below zero for the first time.}. So, it's not entirely inconceivable that negative valuation rates will eventually be used to price and sell IAs in the U.S.. But even if that scenario is farfetched, these two opposing trends raise a critical question. {\em Is it viable to sell refundable IAs in a low enough interest rate environment?} Our answer and key result is that there exists a positive lower bound to the underlying rate, below which cash-refund IAs can {\bf no longer} be offered for sale. The type of income annuities that consumers actually want to purchase simply can't exist. 

\subsection{Outline of the paper, and key results}
\label{sec:outline}

The remainder of this paper is organized as follows. In the Section \ref{sec:literature} we position our paper within the broader annuity literature, with particular link to actuarial science, pension finance and insurance economics. Then, in Section \ref{sec:Analysis} we introduce and formally present the fundamental valuation relationship or algorithm for {\em cash refund} (hereafter CR) and {\em instalment refund} (hereafter IR) income annuities. We begin that section with basic {\em life only} (hereafter LO) income annuities to set notation and terminology and then present a variety of calibrated numerical values for LOIA, IRIA and CRIA assuming Gompertz mortality values. 

\vspace{0.1in}

We should make very clear at this early juncture that our entire paper is predicated on (i.) deterministic mortality hazard rate, and (ii.) the {\em law of large number} which underlies the actuarial pricing of mortality and longevity contingent claims. In other words, we assume the insurance company that manufactures the IAs sells a sufficiently ``large'' number of policies to eliminate all {\em idiosyncratic} mortality risk and that there is no {\em systematic} mortality risk to worry about. Follow-up research will deal with ``small'' portfolios, stochastic mortality hazard rate and the ruin probabilities associated with refundable IAs. 

\vspace{0.1in}

After that, Section \ref{sec:sensitivities} examines the sensitivity of loaded and unloaded consumer prices to both mortality rates and interest rates, which leads to durations and convexity. Once again (counterintuitively) as a result of the recursiveness, we distinguish between a classical measure of Macaulay duration, that is the time-weighted average of discounted cash-flows, versus a (log) derivative with respect to valuation rates. In the case of {\em refundable} IAs these two measures differ from each other, echoing the concept of option-adjusted duration. Next, section \ref{sec:sensitivities} also explores the {\em money's worth ratio} of these products. 

\vspace{0.1in}

Having set the stage with numerical tables and figures, Section \ref{sec:theorems} contains the precise statements of our general theorems and formulae, along with their proofs.  Section \ref{sec:conclusions} concludes the paper proper. The appendix Section \ref{sec:appendix}, contains a brief R-script that can be used to value cash-refund income annuities using a simple bisection algorithm which rapidly converges on the number that equates the actuarial present value to the price. 

\vspace{0.1in}

Before we move-on to the section containing the literature review, we briefly summarize the main theoretical -- and what we believe to be are the novel -- results offered in this paper, albeit much less formally:

\begin{itemize}

\item As long as valuation discount rates denoted by $r$ are positive, there is always a unique price $a^\star$ at which an unloaded CRIA is viable. Naturally, if $r$ is negative, the buyer can never guarantee their money back, although a basic {\em life only} IA is always viable. When $r>0$ the price of the CRIA is decreasing in $r$, and decreasing in the age $x$ of issue, assuming mortality hazard rates increase with age.

\item But, if $r$ exceeds a certain level, there is a unique price $\widehat{a}^{\star}$ at which the CRIA with a given loading is viable. Below a threshold level of $r$, the loaded CRIA fails to be viable. Once again, assuming that mortality hazard rates rise with age, the loaded CRIA is viable for $x$ below a certain level, but not for larger ages. That number is in the mid-70s or early 80s, given current interest rates and typical insurance loadings.

\item Regardless of $r$ and the insurance loading, there is a unique price $a^\circ$ or $\widehat{a}^\circ$ at which an {\em instalment refund} IA is viable. It is decreasing in $r$, once again assuming an increasing hazard rate, and declining in age $x$. Typically this price $\to 0$ at advanced ages in the unloaded case, but not under loading. 

\end{itemize}

\section{Position within the literature}
\label{sec:literature}

This paper is interdisciplinary in nature and overlaps with three distinct research areas of expertise: (i.) actuarial science, (ii.) pension finance and (iii.) insurance economics. In this section we review how our paper engages with the scholarly literature in those three fields.

\subsection{Actuarial Science}
\label{sec:acsci}

Within the literature on actuarial science, the valuation and pricing of income annuities is a topic with three centuries of history, starting with the original paper by Edmund Halley published in \citet{H1693}, in which he introduced the idea of discounting for both mortality and interest and applied it the valuation of income annuities on multiple lives. In fact, by the late 18th century and certainly by the mid 19th century, the valuation and pricing of income annuities under a given (1.) mortality table and (2.) interest rate curve, was viewed as fully settled and secure by insurance actuaries\footnote{See the collection of historical articles contained in \citet{HS1995}.}. Insurance companies as well as local and national governments issued immediate, deferred and reversionary (triggered by death) annuities, all well represented in Victorian novels and literature\footnote{See \citet{P1997, P2014} as well as the popular article by \citet{M2014}.}. The actuarial curriculum for valuing income annuities has only seen minor changes during the last half-century.

\vspace{0.1in}

However, most if-not-all income annuities sold by private insurance companies and voluntarily purchased by individual consumers during the last three centuries did {\bf not} offer the money-back or cash refund feature that has become exceedingly popular in the last decade, and is the focus on this paper. Recall from Table \ref{salesdata} that a cash-refund annuity has become the most dominant rider added to income annuities. So, while historically one could always purchase an income annuity with a guarantee period of 5, 20, 15 or 25 year's period certain, those had no meaningful impact on the valuation and pricing methodology. Mathematically, the period certain simply modifies survival probability from an integrand that is less than 1, to a constant equal to 1, during the period of income certainty. And, even if the survival curve used for pricing depends on the guarantee period selected, perhaps due to anti-selection, this can easily be incorporated into the basic valuation equation by valuing a deferred income annuity together with a term certain annuity\footnote{See \citet{FP2004}, where they find selection effects based on guarantee periods.}.

\vspace{0.1in}

In contrast to these rather trivial integrals, the cash-refund income annuity generates a non-trivial and recursive relationship which embeds the initial value and price into the payout function. In fact as we noted earlier and will carefully demonstrate in Section \ref{sec:theorems}, this money-back guarantee not-only creates mathematical complexity, it might actually lead to no feasible solution. There might not be a fixed point to the valuation equation, or to put it bluntly the product will not exist if insurance loadings (profits, commissions, expenses, or risk capital) increase beyond a certain upper threshold, or valuation rates decline beyond a lower bound.

\vspace{0.1in}

Now, the classic actuarial textbooks are well-aware of the above-noted problem or difficulty with valuing the cash-refund feature of income annuities, and do make reference to an iterative procedure or scheme that is required to solve for the price of the longevity-contingent claim. But, that discussion is often relegated to an obscure appendix or perhaps a page or two of extra problems, whereas entire chapters are dedicated to the valuation of conventional income annuities\footnote{For example, the second edition of the 800 page actuarial textbook by \citet{BGHJN1997}, dedicates a page to cash refund annuities, as an exercise on page 536. The 480 page textbook by \citet{DHW2009}, mentions them briefly as an exercise on page 139, and the 370 page textbook by \citet{P2006} has a page or two devoted to them, abstractly. Needless to say, these textbooks use traditional actuarial notation, symbols and expression that are likely to be inaccessible to economists interested in valuing income annuities. The book by \citet{CT2008} mention value-protected annuities, a form of CRIA, but they don't mention the valuation problems}. More importantly, and as far as the academic contribution of this paper is concerned, we are unaware of any actuarial textbook or prior article that considers the conditions under which the loaded cash-refund income annuity price might not exist.

\vspace{0.1in}

We suspect the reason for the rather limited (historical) attention to the cash refund feature is that they simply weren't popular, and perhaps were even unavailable. It is our understanding that cash-refunds came into existence around the time of the introduction of higher estate taxes (in the US, in the 1940s and 1950s), where beneficiaries preferred a lump-sum payment at death to cover tax liabilities for illiquid assets. We have been told that (old) annuity company accounting systems were originally designed for ongoing and periodic cash-flows, not lump-sum death benefits, which created further impediments to offering cash-refund income annuities\footnote{Source: Author conversations with insurance executives, and in particular Gary Mettler.}. Needless to say, we don’t want to digress into the business history of the income annuity market in the US, but simply note the above background to provide context for why actuarial science might have neglected the cash-refund income annuity as part of its foundational pedagogy. Our paper fills that gap.

\vspace{0.1in}

To be clear, there are recent papers in the actuarial and insurance literature that examine the valuation, hedging and risk management of guarantees embedded within modern 21st century {\em variable annuities} with guaranteed living benefits (GLBs), which are akin to financial put options struck on life or death\footnote{See for example \citet{XCCC2018}, \citet{HZK2017}, \citet{MB2016}, \citet{SM2015}, \citet{NS2011}, \citet{MS2006}, and as well \citet{BKR2008} which examines the many dimensions of variable annuity guarantees.}. The variable annuity market is a multi-billion dollar industry in the US, and many orders of magnitude larger than the market for income annuities. But, as their name suggests, those papers and authors are focused on purely variable products in which separate-account investment returns fluctuate with markets and interest rates. The nature of those guarantees and the risks associated with those derivatives are quite different. And, while those products generate their own mathematical complexities -- and in some circumstances might also fail to be viable -- the focus of our paper once again is on the income annuity at the core of retirement planning. That is the annuity Edmund Halley, Richard Price, Benjamin \citet{G1825}, and Jane Austen would recognize. 

\vspace{0.1in}

In sum, for researchers interested in actuarial and insurance pricing, we contribute to the literature by showing how a simple cash-refund on an income annuity can lead to challenging and interesting mathematical problems. Practically speaking we can help explain why inflation-adjusted cash refund annuities no longer exist; real rates are simply too low. 

\subsection{Pension Finance}
\label{sec:pensions}

The second of the three strands of literature with which this paper engages, has to do with the so-called {\em money’s worth ratio} (MWR) of income annuities, and the competitiveness and efficiency  of the retail income annuity market. In the early 1990s, a number of well-known pension economists in the U.S. introduced the MWR in series of books and papers\footnote{See \citet{FW1990}, as well as the article by \citet{MPWB1999} and the subsequent book cited as \citet{BMPW2001}.}. They suggested using the relative ratio of market annuity prices to theoretical annuity model values (i.e. the MWR) as a measure of price efficiency as well as a way of detecting anti-selection by healthier annuitants\footnote{See \citet{FP2014} as well as the earlier noted \citet{FP2004}.}. The early research indicated that this ratio or number was between 80\% and 90\%, depending on the mortality tables and term structure of interest rate used to compute theoretical model prices. And, when annuitant (versus population) mortality and risk-free rates were used to compute the MWR, the ratio often exceeded 100\%. This was taken as a sign the market was efficient and functioning properly. Recent research work updating the original estimates continue to point to MWRs in the same range, despite the decline in both interest rates and mortality rates during the last three decades\footnote{Recent papers using the MWR are \citet{PS2021}, \citet{BFN2021} as well as \citet{WMHG2021}.}.

\vspace{0.1in}

The economic impetus and relevance of these MWR studies revolve around the global demise of traditional Defined Benefit (DB) pensions and the trend towards individual and portable Defined Contribution (DC) investment accounts. Recall that the default option in a DB plan is retirement income for life, which obviously resembles the payout of an income annuity. In many cases, it not only is a {\em default} option in the pension plan, it’s the {\em only} option available at retirement. For example, the income annuity provided by U.S. Social Security (or the Canadian version, CPP) cannot be taken as a lump sum in cash. In contrast, the default option in a DC plan is either a lump-sum of money at retirement or perhaps a systematic withdrawal plan with a default spending rate. Neither of those options protects the retiree against longevity risk, i.e. the risk of outliving resources. That is the main reason why income annuities have received such attention from economists over the last few decades. From a retirement policy perspective, income annuities are being promoted by both academic economists as well as insurance company executives to help replace the lost longevity insurance in DC plans. Indeed, recent US legislation around so-called safe harbours and menu options has accelerated this trend\footnote{See the SECURE Act and associated policy issues, as well as the original article by \citet{B1990} discussing the connection between pensions as a form of retirement income insurance.}.

\vspace{0.1in}

And yet most, if not all, of the published research articles and policy papers that compute or reference these MWRs as a measure of cost efficiency have focused on historical data, pricing information and quotes for {\em life only} income annuities without any cash-refund features. Just in the year 2021 alone we have come across a handful of articles that compute MWRs, all focusing exclusively on life only income annuities. This isn’t just a matter of empirical convenience or the availability of historical data. Most of the highly cited theoretical papers in the annuity allocation literature, i.e. deriving optimal dynamic strategies around annuitization, have also focused on {\em life only} annuities\footnote{Starting with \citet{Y1965}, as well as \citet{D1981}, and more recently \citet{KT2005}, \citet{MY2007}, \citet{HMMS2009}, \citet{ILM2011}, \citet{CHMM2011}, \citet{KNW2011}.}. Any reference to a guaranteed period for the income annuity is usually noted within the context of it being a fixed-income bond, and thus irrelevant to the financial economics of longevity risk pooling. 

\vspace{0.1in}

And yet, we repeat once again that those are not the income annuities consumers are buying. Rather, most of them want, demand and pay extra to receive their money back over time. Our main point here is that the MWR literature must progress to using cash-refund pricing algorithms, and our paper offers encouragement and a path forward to those researchers. In addition to our theoretical results, the appendix to this paper contains a simple bisection-based algorithm (in R) that can be used to properly value cash-refund income annuities under a Gompertz law of mortality. The mortality curve can easily be discretized, modified or flattened. 

\subsection{Insurance Economics}
\label{sec:economics}

The third and final strand of literature to which this paper contributes, relates more generally to the so-called {\em annuity puzzle} and the very low levels of voluntary annuitization observed among retirees, despite the classic \citet{Y1965} result which argues for full annuitization\footnote{For the initial statement and framing, as well as more recent empirical evidence on the extent of the puzzle, and reasons a thin annuity market might be rational, see \citet{M1988}, \citet{DBD2005}, \citet{BKMW2008}, \citet{BPT2011}, \citet{P2013}, \citet{BL2014}, and finally \citet{RS2015}.}. The annuity puzzle literature, which is just as vast as the above noted pension economic or actuarial science work, has offered a variety of both classical (market frictions and incompleteness) and behavioural (anchoring, mental accounting, loss aversion, etc.) reasons for why consumer avoid annuitization. It continues to be an active area of research.

\vspace{0.1in}

{\em Prima facie}, the fact that from a very small group who actually choose to annuitize, an even smaller group select the life only option further exacerbates the puzzle. It's a simple actuarial fact that per premium dollar, the pure {\em life only} income annuity provide the highest level of income, the greatest mortality credits and longevity pooling benefits to the annuitant. And yet, recall that in Table \ref{salesdata}, a mere 10\% of annuitants selected the life only income annuity. A full 90\% of buyers opted for some form of guarantee, and within that category a majority of those elected a cash-refund. 

\vspace{0.1in}

In fact, to add yet another layer of puzzlement, as interest rates have declined over the last decade or so, the payout from income annuities have declined across the entire spectrum of riders. The present value of everything including the cost of retirement income goes up. Still, consumers have shifted towards buying even more expensive income annuities that reduce their cash-flow and income even further. At the risk of overusing the word, we pose this as a further puzzle. Refunds are more expensive than ever, and continue to increase in popularity. 

\vspace{0.1in}

But, the valuation model and results described in this paper allow us to shed light on this matter by highlighting and isolating the embedded life insurance component within the cash-refund annuity. As such, the demand for cash-refund income annuities might be driven not-only by behavioural considerations noted above, but perhaps by rational attempts by the annuitant (buyer) to signal to the insurance company (seller) that they are not as healthy as a conventional annuitant, purely interested in protecting against longevity risk. It also might assist in maximizing a bequest value, more effectively compared to a period certain\footnote{See, for example, the article by \citet{S2008} as well as the follow-up book referenced as \citet{S2010}, in which he examines what he calls {\em refundable annuity options}, albeit with a similar rationale.} option. Stated differently, the cash refund income annuity might be a better {\em deal} for consumers, since it might reflect a lower level of anti-selection. 

\vspace{0.1in}

Furthermore, recent experimental evidence suggests that money back guarantees are more likely to be selected as part of a so-called {\em cushion effect.}\footnote{See the paper by \citet{K2016}, who makes an argument consistent with the IA riders that consumers select, and see \citet{BBM2020} for additional survey-based evidence.} And, while that would be an empirical question, we will return to this in the conclusion of the paper when we discuss further research. For now, we continue to Section \ref{sec:Analysis} where we introduce and derive the valuation and pricing technology.

\section{Actuarial analysis}
\label{sec:Analysis}

\subsection{Notation and terminology}
\label{sec:notation}

In this paper we have tried as much as possible to adhere to standard actuarial notation, but have deviated in a few places (elevating subscripts and/or introducing superscripts) to reduce clutter and crowding.

\begin{itemize}

\item $a(x,r)$ is the price at age $x$, of a single premium life only (LO) income annuity (IA) paying $\$1$ in continuous time, under a valuation rate $r$, and with no loading. Because of the absence of loading, it agrees with the actuarial present value of the LOIA. In other words, it is the present value of the payment stream, discounted for both interest rates and mortality. Or, by the law of large numbers, the cost per contract of hedging a large number of identical such annuities. We will (quite often) abbreviate it by $a(x)$, or occasionally just $a$, when the interest rate $r$ and/or age $x$ are clear from the context. 

\item ${a^{\star}}(x,r)$ is the price at age $x$, under a rate $r$, of an (unloaded) IA paying $\$1$ in continuous time, {\bf and also} guaranteeing a cash-refund (CR) at death of: ${a^{\star}}(x,r)-t$, if death occurs at $t \leq {a^{\star}}(x,r)$. Because there is no loading, it is also the actuarial present value (or {\em value}) of this annuity. Naturally: ${a^{\star}}(x,r)>a(x,r)$. Once again, we might use the simpler $a^{\star}(x)$ or just $a^{\star}$ when the context is clear.

\item $a^{\circ}(x,r)$ is the price at age $x$, under a rate $r$, of an (unloaded) IA paying $\$1$ in continuous time, guaranteeing at least $a^{\circ}(x,r)$ years of payments to the annuitant or beneficiary. Because there is no loading, it is also the actuarial present value (or {\em value}) of this annuity. This is the instalment refund income annuity (IRIA), which is effectively a life-only IA, with an additional $a^{\circ}(x,r)$ years of guaranteed payments. In Section \ref{sec:theorems} we prove that $a^{\ast}(x,r)>a^{\circ}(x,r)$ at any age $x$ and any interest rate $r>0$, and in the current section focus on valuation and numerical examples.

\item An LOIA is {\it $\tau$-period certain} if the \$1 paid continuously lasts for life or time $\tau$, whichever is greater. We denote its price as $a(x,r,\tau)$. This allows us to express an IRIA by adjusting $\tau$ recursively, that is, $a^{\circ}=a^{\circ}(x,r)$ satisfies $a^{\circ}=a(x,r,a^{\circ})$. Alternatively, $a(x,r,\tau)$ can be realized as the present value of $\$1$ for time $\tau$ (i.e. $\frac{1-e^{-r\tau}}{r}$) plus a life annuity that only starts payment at time $\tau$. The latter, called a {\it deferred} or {\it delayed} IA and abbreviated as DIA, has been the focus of intense interest among pension economists and has recently been granted special tax treatment in the U.S.\footnote{See \citet{GW2010}, \citet{P2016}, \citet{HMY2017}, \citet{AG2019}, \citet{HMM2020}, \citet{HHMNS2020}, and the regulatory exemptions granted to Qualified Longevity Annuity Contracts (QLACs)}.

\item $\pi \geq 0$ denotes the insurance loading or mark-up from {\em value} to {\em price} which is applied to the IA by the insurance company, mainly to cover profits and expenses, although in practice $\pi$ would include a safety margin to protect against adverse mortality deviations. Practically speaking, loading a pure actuarial premium by $\pi$, results in a present value of lifetime cash-flow of: $C=\frac{P}{(1+\pi)a}$, where $P$ is the premium, or the annuitized wealth. In other words, the price $\widehat{a}$ of a loaded LOIA will simply $=(1+\pi)a$. The relationship between loaded and unloaded prices will be more complicated in the case of a CRIA or an IRIA.

\item $A(x,r)$ is the price at age $x$, of a single premium {\bf insurance} policy paying $\$1$ at death, under a valuation rate $r$. This implies that $A(x,r) +r a(x,r) = 1$, which can be shown via {\em No Arbitrage} arguments. Rearranging this relationship leads to: $\frac{1}{a(x,r)}-r=\frac{A(x,r)}{a(x,r)} > 0$, which is deemed to be {\em the mortality credits} at age $x$.

\item Finally, $T_x$ denotes the remaining lifetime random variable whose density and distribution function are denoted $f_x(t)$ and $F_x(t) = \Pr[T_x \leq t]=1- \, _tp_x$. The conditional survival probability under a Gompertz law, which we will use for all our numerical examples, is: $\, _tp_x = e^{-\int_x^{x+t} \lambda_s \, ds}=\exp\{e^{(m-x)/b}(1-e^{t/b})\}$, where $m$ is the modal value and $b$ is the dispersion coefficient, and the mortality hazard rate is: $\lambda_s=(1/b)e^{(s-m)/b}.$

\end{itemize}

\subsection{Life Only IA}
\label{sec:LOIA}

Following standard actuarial methodology, 

\begin{equation}
a(x,r) \; = \;  \int_{0}^{\infty} e^{-rt} \Pr[T_x \geq t] dt \; = \; 
\int_{0}^{\infty} e^{-rt} \, (_tp_x) dt,
\label{LOIA}
\end{equation}
Under Gompertz mortality, this can be solved analytically\footnote{For example, see \citet[Ch. 6]{M2006}.} and leads to the expression:
\begin{equation}
a(x,r)  \; = \; \frac{ b \Gamma(-rb,e^{(x-m)/b}) }{ \exp\{(m-x)r-e^{(x-m)/b}\} },
\label{GAVM}
\end{equation}
where $\Gamma(A,B)$ is the incomplete Gamma function. For calibration and comparison purposes we now display values for LOIA prices assuming two different interest (valuation) rates: $r=2\%$ and $r=4\%$. We provide a short R-script in the technical appendix, which comes from integrating equation (\ref{LOIA}), and later will be compared against the instalment refund IA and cash refund IA numerical values. Table \ref{LOIA_numbers} provides a range of numerical values of $a=a(x,r)$, for $x=55,x=65,x=75$ under valuation rates $r=2\%$ and $r=4\%$. Likewise, Figure \ref{FIG1} displays the LOIA price as a function of valuation rates $r$, at ages $x=65,75$.

\begin{center}
\fbox{{\bf Table \ref{LOIA_numbers} Placed Here}}
\end{center}

Figure \ref{FIG1}a (left panel) and Figure \ref{FIG1}b (right panel) display the usual and expected declining pattern for $a$, as valuation rates are increased. We will return to this Figure when we discuss the cash-refund $a^{\star}$ and instalment refund $a^{\circ}$ values.

\begin{center}
\fbox{{\bf Figure \ref{FIG1} Placed Here}}
\end{center}

\subsection{Cash Refund IA}
\label{sec:CRIA}
We now arrive at the ``star'' of this paper, and the focus of our attention, the cash-refund income annuity (CRIA). The first thing to note is that its price must be defined recursively, since the periodic cash-flow itself depends on the original price. In the absence of loading, this can be expressed mathematically as follows: 
\begin{equation}
{a^{\star}}(x,r) \; = \;  \int_{0}^{\infty} e^{-rs} \, (_sp_x) ds
\; + \; \int_{0}^{{a^{\star}}(x,r)} \left({a^{\star}}(x,r) \, - \, s \right) e^{-rs} \, (_sp_x) \, \lambda_{(x+s)} \, ds,
\label{astar}
\end{equation}
where the right hand side represents the actuarial present value of payments. 
The annuitant receives $\$1$ for life, but if they die early, that is before the entire ${a^{\star}}(x,r)$ has been returned, the beneficiary receives a (type-of) life insurance payment consisting of the difference between the price ${a^{\star}}(x,r)$ and payments received prior to death. The price is constructed in dollars but also denotes an important time metric. Death before age $y=x+{a^{\star}}(x,r)$ triggers a (declining) death benefit. 

\vspace{0.1in}

Notice how the first integral in equation (\ref{astar}) is the basic income annuity price $a(x,r)$, and the second integral represents the non-negative life insurance component. The CRIA price minus the LOIA price can also be expressed as:
\begin{equation}
{a^{\star}}(x,r) - a(x,r) \; = \; \int_{0}^{{a^{\star}}(x,r)} \left({a^{\star}}(x,r) \, - \, s \right) e^{-rs} \, (_sp_x) \, \lambda_{(x+s)} \, ds
\label{cria_premium}
\end{equation} 
The technical appendix contains a short bisection-based algorithm for locating numerical solutions to a modified version of this equation. For example, under a Gompertz law of mortality with $m=90,b=10$ and $r=3\%$, the LOIA price at age $x=65$ is ${a^{\star}}(65,0.03)=15.25$ while the CRIA price is ${a^{\star}}(65,0.03)=16.91$, both per dollar of lifetime income. Note, once again, how in contrast to equation (\ref{GAVM}), the $a^{\star}$ appears on both sides of equation (\ref{cria_premium}) which leads to an implicit recusivity, and possibly to no solution at all. Section \ref{sec:theorems} offers a proof of existence and uniqueness, and shows that $a^{\star}(x,r)$ actually declines in both $x$ and $r$. 

\vspace{0.1in}

Nevertheless, at the age of $x=65$, under a $r=3\%$ valuation rate, the difference between the CRIA and LOIA price is $1.66$ per dollar of lifetime income, which is about $11\%$ more expensive. That is the cost of the money-back guarantee, and the extra premium that buyers are paying for this rider. Additional numerical values are presented in Table \ref{CRIA_numbers}, where they can be compared with the LOIA prices noted in Table \ref{LOIA_numbers}. Once again, the cash-refund feature is obviously a costly rider, and even more so in absolute as well as relative terms, as valuation rates decline, (e.g. 2\% versus 4\%).

\begin{center}
\fbox{{\bf Table \ref{CRIA_numbers} Placed Here}}
\end{center}

In Section \ref{sec:theorems}, we will show that \eqref{astar} can be reformulated as the following;
\begin{equation}
\int_{{a^{\star}}(x,r)}^\infty e^{-rt}{}(_tp_x)\,dt=\int_0^{{a^{\star}}(x,r)} e^{-rt}r({a^{\star}}(x,r)-t){}(_tp_x)\,dt.
\label{twoaccountversion}
\end{equation}
A particularly useful way to think about equation (\ref{twoaccountversion}), as well as the financial process by which the insurance company manages the CRIA payout, is as follows. Initially, the annuitant's entire premium ${a^{\star}}(x,r)$ is placed into a so-called payout or {\em phase-one} account, which only contains the sum that is refundable at death and is used to generate and payout the annuity income up to time ${a^{\star}}(x,r)$. But, all interest income earned within this (shrinking) {\em phase-one} account goes into -- and is accumulated in -- a second {\em phase-two} account. Eventually, the insurer uses that {\em phase-two} account to pay lifetime income after the {\em phase-one} account is emptied or depleted and the formal guarantee period is over.  Of course, insurance companies don't actually maintain two accounts for every annuitant, and in reality the entire premium goes into one large general account, but it helps to think in this bifurcated manner. With this framework, the LHS of equation (\ref{twoaccountversion}) is the discounted value of the {\em phase-two} payments. The RHS of equation (\ref{twoaccountversion}) is the income generated by the {\em phase-one} account, since $({a^{\star}}(x,r)-t)$ is the account balance, and ${}_tp_x$ is the fraction of (initial) accounts surviving.

\subsection{Instalment Refund IA}
\label{sec:IRIA}

A closer (and cheaper) cousin to the cash refund annuity (CRIA) is the installment refund annuity (IRIA), under which payments continue to the beneficiary until the entire premium is returned. Unlike the CRIA, there is no lump-sum death benefit, rather payments continue {\em as if} the annuity included a period certain of $a^{\circ}(x,r)$ years. Going back to the fundamental valuation expression, payments for the first $t \leq a^{\circ}(x,r)$ years are only discounted for interest (not mortality), and payments after $t > a^{\circ}(x,r)$ are discounted for both interest and mortality. Mathematically this can be expressed as:
\begin{align}
a^{\circ}(x,r) \;  &=  \;  \int_{0}^{a^{\circ}(x,r)} e^{-rs} ds
 +  \; \int_{a^{\circ}(x,r)}^{\infty} e^{-rs} \, (_sp_x) \, ds   \nonumber \\ 
 &  =
 \begin{cases}
 \frac{1}{r} (1-e^{-r a^{\circ}(x,r)})
 +  \; \int_{a^{\circ}(x,r)}^{\infty} e^{-rs} \, (_sp_x) \, ds, & r > 0, \\
 a^{\circ}(x,r)
 +  \; \int_{a^{\circ}(x,r)}^{\infty} \, (_sp_x) \, ds,  & r  = 0.
\label{alpha}
\end{cases}
\end{align}
IRIA values under the same parameters as before are shown in Table \ref{IRIA_numbers}. Equation (\ref{alpha}) is recursive in $a^{\circ}(x,r)$, but unlike the CRIA price ${a^{\star}}(x,r)$ it is mathematically simpler. One can immediately see that under $r=0$, $a^{\circ}(x,r)=\infty$, and the IA isn't viable. 

\begin{center}
\fbox{{\bf Table \ref{IRIA_numbers} Placed Here}}
\end{center}

Here is a scheme for locating $a^{\circ}(x,r)$, if you have a convenient or closed-form expression, such as the Gompertz formula in equation (\ref{GAVM}), but for an IA with a $\tau$-year period certain (PC). First, note that $a(x,r,\tau)$ is increasing in $\tau$, and is $>\tau$, when $\tau=0$. But, $a(x,r,\tau) \rightarrow \frac{1}{r} < \tau$ as $\tau \rightarrow \infty$, as long as $r>0$. So, increase $\tau$ until the curves cross and $a(x,r,\tau) = \tau$. In Section \ref{sec:theorems} we will use a related argument to also show uniqueness, and to handle loadings. 

\subsection{Insurance loading by $\pi$}
\label{sec:loadings}

Now, let's return to the subject of insurance loadings. In contrast to the LOIA computed in the prior subsection, it would be inappropriate and technically incorrect to multiply ${a^{\star}}(x,r)$ by an insurance loading of $(1+\pi)$ to convert them from unloaded to loaded prices. Why? Because by charging or adding the extra loading, the insurer is obligated to pay a greater death benefit, and for longer, even though their overall obligation to pay \$1 for life hasn't changed. Using the above terminology, the {\em phase-one} account must last longer. And, to preempt our main result, sometimes this isn't possible. 

\vspace{0.1in}

Let $\pi \geq 0$ denote the proportional insurance loading added to the actuarial present value of a CRIA to yield its price. If $\widehat{a}^{\star}$ denotes the CRIA price, then the actuarial present value is $\frac{\widehat{a}^\star}{1+\pi}$,  so $\widehat{a}^{\star}$ must satisfy: 

\begin{equation}
\frac{\widehat{a}^{\star} }{1+\pi}= a
\; + \; \int_{0}^{\widehat{a}^{\star}} \left(\widehat{a}^{\star} \, - \, s \right) e^{-rs} \, (_sp_x) \, \lambda_{(x+s)} \, ds.
\end{equation}
For given values of $(x,m,b,\pi,r)$ it's unclear a solution exists. And, as a consequence of the results in Section \ref{sec:theorems}, assuming Gompertz mortality with parameters $m=90$ and $b=10$, you can {\bf not} sell a cash-refund immediate annuity (CRIA) if the valuation rate ($r$) falls below a certain level. Discounted expectations won't exist. The lowest viable valuation rate under which the Cash Refund Income Annuity (CRIA) is still feasible will be denoted by $r_{\pi, x}$, expressed in basis points -- see Table \ref{viabler}. Thus, for example,  with Gompertz Mortality ($m=90$ and $b=10$) an inflation-adjusted CRIA would not be feasible at $x=75$, under a loading of $\pi=15\%$, if real rates used for pricing fall under $r_{0.15, 75}=1.01\%$.

\begin{center}
\fbox{{\bf Table \ref{viabler} Placed Here}}
\end{center}

Loaded prices for IRIAs are defined in the same way, so we have expressions (or at least algorithms) for all three IAs. Figure \ref{FIG2} displays LOIA, IRIA and CRIA prices assuming Gompertz mortality with $(m=90,b=10)$, and a valuation interest rate of $r=2\%$. The dashed curves capture the prices without loading, while the solid lines show the loaded prices. At younger ages (visually under age 40) all three prices are nearly identical and the guarantees or refunds don't make much of a difference in value. But, at older ages the IRIA and CRIA are (relative) more expensive. 

\vspace{0.1in}

Notice how at the age of $x=60$ the IRIA and CRIA prices are approximately 10\% higher than the LOIA prices. But, at age $x=75$ and beyond, when mortality credits start to ``kick in'', the premium for refunds can exceed 50\% to 70\% of the LOIA price. And at higher ages, the CRIA price starts increasing with $x$, and is eventually not viable.

\vspace{0.1in}

Observe that ${a^{\star}}(x,r)>a^{\circ}(x,r)$ and ${\widehat{a}^{\star}}(x,r)>\widehat{a}^{\circ}(x,r)$ here.  We will see in Section \ref{sec:theorems} that these relations always hold. This is plausible in the unloaded case, but seems much less obvious under loading. Note also that while the two LOIA prices agree asymptotically as $x\to\infty$, in other words whether the annuity is loaded for costs and profits, or sold with no loading at all, the limiting price at advanced ages is the same. Yet, this is not the case for IRIAs and we explore this fully in Section \ref{sec:theorems}.

\begin{center}
\fbox{{\bf Figure \ref{FIG2} Placed Here}}
\end{center}

\section{Sensitivity, duration, and money's worth ratio} 
\label{sec:sensitivities}

\subsection{Impact of age and rates}
\label{sec:duration}

In this section we discuss the ``sensitivities'' of the life-only (LO), instalment refund (IR) and cash-refund (CR) income annuity (IA) prices, relative to changes in the underlying parameters age $x$ and valuation rate $r$. Recall from section \ref{sec:Analysis} that although we have a closed-form analytic expression for the LO prices $a(x,r)$ in a Gompertz framework, the IR price $a^{\circ}(x,r)$ and CR price ${a^{\star}}(x,r)$ can only be computed via iteration. Those iterations also lead to a rather complex process for computing partial derivatives, which among other interesting quantities are the core ingredient of so-called duration and convexity calculations. 

\vspace{0.1in}

To complicate matters further, although for LO the loaded price $\widehat{a}=(1+\pi)a$, is a simple multiple of the unloaded price $a$, that is not the case for the other two annuities. For IRIA and CRIA the loaded prices, which we denoted in Section \ref{sec:Analysis} by $\widehat{a}^{\circ}$ and $\widehat{a}^{\star}$, are no longer linear multiples $(1+\pi)$ of the $a^{\circ}$ and ${a^{\star}}$. Therefore ``sensitivities'' will not scale either, and if $r_{\pi,x}$ exceeds $r$ they might not exist at all, for the CRIA.

\vspace{0.1in}

Figure \ref{FIG3} displays the sensitivities of IA prices (loaded and unloaded) as a function of age. Notice in panel (a) that (in the unloaded case), the sensitivities $\frac{\partial}{\partial x}$ of all three IAs are consistently negative, consistent with prices declining in age $x$. Interestingly, notice how prior to the age of (approximately) 85, aging reduces the LO price at a greater rate than the CR price (the derivative is negative and lower). But after that age the situation is reversed and the reduction in the LOIA price slows down while the CRIA reduces at faster rate. In contrast to the story told in panel (a), the age derivative of the loaded CRIA price $\widehat{a}^{\star}$ in panel (b) is consistently above that of the loaded LOIA price $\widehat{a}$. In fact, and this once again ties into our earlier result, at some age the derivative goes positive. Once again, this means that older people will have to pay more (not less) for the same dollar of lifetime income.

\begin{center}
\fbox{{\bf Figure \ref{FIG3} Placed Here}}
\end{center}

Turning to the impact of interest rates $r$, we wish to know which of the three IAs has the largest or greatest percentage sensitivity to a change $\Delta r$ in interest rates. Or, in the language of bond pricing, at any given age, whose duration is highest? While we measured the sensitivity to age $x$ as the straight partial derivative $\partial / \partial x$, as discussed above, for valuation rates $r$ we will proceed in a manner that is analogous with the bond literature. The new (key) symbol is $Dr[\cdot]$, which will operate on either the loaded or unloaded prices. So, from here on, will focus attention on the following mathematical expression:

\begin{equation}
Dr[a]:=\frac{-\partial a / \partial r}{a}, 
\label{AnnuityDuration}
\end{equation}
and similarly for the other IAs. The above expression in (\ref{AnnuityDuration}) has been labeled {\em life annuity duration} in the insurance literature\footnote{See, for example, \citet{CKM2016}.} and we adhere to that definition in this paper. Its convenience and use derives from the well-known approximation that:
\begin{equation}
\frac{a(x,r+\Delta r)}{a(x,r)}-1= -Dr[a(x,r)] \times \Delta r
\end{equation}
When applied to the standard life-only IA $a$, the life annuity duration $Dr[a]$ ``operator'' will collapse to the well-known Maculay duration -- a weighted and discounted time average of cash flows. In fact, the loading $(1+\pi)$ will cancel and the $Dr[a]$ is identical for the loaded and unloaded prices. However, the $Dr[\cdot]$ will not correspond with the Macaulay duration in the case of the IRIA and CRIA prices, once again due to the same culprit, recursivity. We will get back to that in a moment, and despite the popularity (and history) of Macaulay, we limit our analysis and comments to duration defined as $Dr[\cdot]$.

\vspace{0.10in}

The rationale for the $Dr[a]$ as our definition of life annuity duration -- even for the IRIA and CRIA -- can be easily understood when it's applied to a simple (no mortality) present value (a.k.a. zero coupon bond) maturing at \$1. In that case $a:=e^{-rT}$ and $Dr[a]=T$, which is evidently the {\bf time} until receipt of the payment. This is why units of $Dr[a]$, and by extension $D\lambda[a]$, are quoted in years. Moreover, if we replace $r$ by $r+\Delta r$, the percentage change in the (new) value of $a$ can be approximated to first order (Taylor expansion of $e^{\Delta T}$) by $\Delta rT$ percent. Now, to continue developing intuition for our life annuity duration, let's take the case of a life only IA in which the force of mortality is constant and lifetimes are exponentially distributed. The IA price is: $a:=\int_{0}^{\infty}e^{-(r+\lambda)s}ds=(r+\lambda)^{-1}$, and from the definition in \eqref{AnnuityDuration}, the $Dr[a]=a$. In this case, the percentage change in the IA can be approximated by $-a \Delta r$, for small enough values of $\Delta r$. For example, if $r=3\%$ and $\lambda=2\%$ (which is a life only IA issued at a life expectancy of 50 years), the IA price is $20$ dollars and the life annuity duration is (also) $20$ years. If we add $+25$ basis points to the interest rate $r$ the new IA price is simply: $19.0476=1/(0.0525)$ which is a decline of $4.76\%$. In contrast, using the life annuity duration approximation, multiplying the $-Dr[a]=20$ by $\Delta r = 0.0025$ is a change of $-5\%$. We would like to obtain similar expressions (and approximations) for the IRIA and CRIA prices.

\vspace{0.10in}

 As for the relationship between $Dr[a]$ and Macaulay duration we offer the following observation. Regardless of the exact IA type, if cash flows from the insurance company to the annuitant at a rate $\phi_t$ at time $t$, the associated Macaulay duration is $\frac{\int_0^\infty t e^{-rt}\phi_t\,dt}{\int_0^\infty e^{-rt}\phi_t\,dt}$, measured in units of time. In other words the arrival times $t$ are weighted by the present values $e^{-rt}\phi_t$ of the associated cash flows. So, in the case of an LOIA, if ${a}(r)=\int_0^\infty e^{-rt}\phi_t\,dt$, Macaulay duration is $-\frac{1}{{a}}\frac{d{a}}{dr}$, which is consistent with $Dr[{a}]$ noted above in \eqref{AnnuityDuration}. However, this is no longer the case for IRIA and CRIA. Macaulay duration, as defined above, is simple to compute -- just hold ${a^{\star}}$ constant in the right hand side of \eqref{astar}, while differentiating with respect to $r$. This no longer equals $Dr[a]$, b/c changing $r$ does {\it not} leave ${a^{\star}}$ constant, so the expression loses its interpretation as a sensitivity to interest rates. Nevertheless, we will continue to report life annuity duration in years.

\subsection{Numerical examples}
\label{sec:examples}

Here we provide some intuition for life annuity duration. For starters, it should be clear from the left-hand panel of Figure \ref{FIG4}, that at younger ages the life annuity durations of all three {\em unloaded} products are virtually indistinguishable from each other. The values of $Dr[a]$, as well as $Dr[a^{\circ}]$ and $Dr[{a^{\star}}]$ decline with age $x$, but much more rapidly for the LOIA. Indeed, the life annuity durations are in the vicinity of 10-15 years around the retirement ages, which is what one might expect from a fixed income product that pays ``on average'' until the period ends. Note that at higher issue ages, and certainly during the traditional period of retirement between age 65 and 85, the life annuity duration of the CR and IA rider is 8-10 years higher (larger, more sensitive) than the LO version. This depends on parameter values (in this case $m=90,b=10,r=2\%$) but is a feature of all refunds at advanced ages. Higher duration has immediate and practical applications for hedging and ALM.

\begin{center}
\fbox{{\bf Figure \ref{FIG4} Placed Here}}
\end{center}

Now, the situation is markedly different and rather counter-intuitive for the loaded IA prices. The right-hand panel of Figure \ref{FIG4} display the $Dr[\widehat{a}]$ for the LOIA, together with $Dr[\widehat{a}^{\circ}]$ for the IRIA and $Dr[\widehat{a}^{\star}]$ for the CRIA. A number of odd things can immediately be observed. First, at younger ages, the life annuity duration for the IRIA and CRIA are now lower and under the LOIA curve. They actually cross over -- for these particular parameter values -- somewhere around the age of $x=40$. After that point the life annuity durations rank in the same relative order as the unloaded prices. But, the oddities continue. While the LOIA duration continues to decline, and is unaffected by the $\pi$ loading which cancels out in the equation (derived in Section \ref{sec:theorems}), the IRIA and CRIA life annuity durations begin to increase somewhere around the age of $x=70$. In fact, the CRIA duration -- again, these are loaded prices -- asymptotes around the critical age at which the value no longer exists. Once again, we refer readers to Section \ref{sec:theorems} to better understand these aspects of the fundamental pricing equation.
 
 \subsection{Money's Worth Ratio}
 \label{sec:MWR}

We conclude this section with a brief discussion of the implications of our work for the computation of the so-called {\em Money’s Worth Ratio} (MWR) which, as we noted in the literature review, is an area of continued interest among scholars in pension economics and finance. Generally speaking the MWR is defined as the ratio of the {\em theoretical model value} and {\em empirical market price} of an IA. If the measured ratio is less than one, which is usually the case since market prices are loaded, then one is said (by an economist) to be paying more than a fair price for the IA. If the ratio happens to be greater than one (which is quite rare) the IA is considered to be a good deal (by an economist) and if the ratio is exactly equal to one, the market price would be perfectly fair. This way or approach to analyzing annuity market prices was first introduced over 30 years ago by the economists \citet{FW1990}, has become the {\em de facto} standard in the literature and continues to be used; see for example \citet{PS2021}. Using our notation one can think of the {\em MWR} as being approximately equal to $\frac{1}{1+\pi}$. 

\vspace{0.1in}

Now clearly the numerator in the {\em MWR} ratio -- the theoretical model value of the IA -- requires an assumption for both mortality and discount rates, both of which might be more complicated than the Gompertz $(m,b)$ law and a constant valuation rate $r$. After all, one never really knows exactly what the insurance company is assuming in their pricing. It is therefore more accurate to present the {\em MWR} estimates as being conditional on a model. So, if $C$ denotes the annual cash-flow generated by the IA and $P$ denotes the market premium or price, then $P/C$ is the {\em empirical} price per dollar of lifetime income, which can be compared with the {\em theoretical} LOIA price $a$ or CRIA price $a^{\star}$. With that in mind we will define the Gompertz {\em MWR} for both the CRIA and the LOIA, naturally as follows:
\\
$$
{\tt MWR}^{\star}(x,r,m,b \mid P,C) = \frac{a^{\star}(x,r,m,b)}{P/C}, \;\;\;\;\;\; {\tt MWR}(x,r,m,b \mid P,C) = \frac{a(x,r,m,b)}{P/C}
$$
Our objective is to estimate the {\em MWR} for LOIA and CRIA, to see if there is any meaningful difference between them. Table \ref{MWRtable} displays results and this is how to interpret those numbers. A 65-year-old male who invested $P=\$100,\!000$ into an IA in early July 2021, would have been entitled to (guaranteed) approximately $C=\$5,\!844$ per year for life, paid monthly, if he selected the life only version (according to data from Fidelity Investments). But had he selected the cash refund the payout would have been only $C=\$5,\!280$ per year, for reasons that should be evident by this point.

\begin{center}
\fbox{{\bf Table \ref{MWRtable} Placed Here}}
\end{center}

Using those numbers, if we assume a Gompertz mortality model with $m=90$ years and $b=10$ years, and a constant $r=2\%$ valuation rate, the {\em MWR} estimate of the life only payout is exactly $0.996$ per premium dollar. To be clear, this number is obtained by dividing the theoretical value $a(65,0.02,90,10)$ by the empirical ratio $\frac{100000}{5844}=17.1157$. Note that we are {\bf not} claiming the insurance company uses this model and this valuation rate to price LOIAs. Rather, under those assumptions the MWR would indicate perfect fairness to an economist. In fact, we iterated $r$ until we converged to {\em MWR=1}. 

\vspace{0.1in}

But here is where it gets interesting. When we computed this ratio using the same $(r,m,b)$ parameter assumptions at age $x=65$, but using the lower CRIA payout of $C=\$5,\!280$, the ${\tt MWR}^{\star}$ estimate is actually higher than 1. In fact, the ${\tt MWR}^{\star}=1.03$ which might not seem like very much, but is actually a rather large number within the money's worth ratio literature. More importantly, all the remaining rows in Table \ref{MWRtable} tell the exact same story. While the ${\tt MWR}(x,r,m,b) \approx 1$ for the LOIA, the ${\tt MWR}^{\star}(x,r,m,b) > 1$ for the CRIA. {\em Why?}  

\vspace{0.1in}

An answer to this question, or a possible explanation for this empirical fact gets back to the economic essence of cash refund immediate annuities. Recall that there were two components or integrals in the valuation of the CRIA price $a^{\star}$. This first was the annuity payment of \$1 for life, or the basic $a$ price. On a stand-alone basis, the relevant mortality assumptions for that component would be (healthy, anti-selected) annuitant mortality.  But the second component in the construction of $a^{\star}$ is associated with a life insurance payment due upon (early) death. It's unclear that portion should be valued and discounted using annuitant mortality, given that it reflects a payment upon death. Moreover, it could very well be that buyers who select a cash refund (versus a life only) feature are in some way signalling they are less healthy than average annuitants. This then induces the insurance company to use less conservative (read: lower than $m=90$) values, which results in a higher than unity ${\tt MWR}^{\star}$, computed relative to the conservative (read: $m=90$) values.

\vspace{0.1in}

In sum, although this is all rather speculative and demands further research, it seems that when measured properly the {\em MWR} values are higher than life only versus cash refund IAs. And, since they now form the majority of sales in the U.S., we would urge the literature to focus more attention on CRIAs and their possible role as a vehicle for satisfying bequest motives in a more cost-efficient manner. 

\section{Theorems and General Formulae}
\label{sec:theorems}

\subsection{Theorems and proofs: CRIA with fair pricing}
\label{sec:fairCRIA}
In this sub-section we provide a formal and proper statement of our theorems and corresponding lemmas, as well as the proofs alluded to in the body of the paper. We start by rewriting equation \eqref{astar}, in terms of the density $f(t)$ of $T$:
\begin{equation}
{a^{\star}}=\int_0^\infty \frac{1-e^{-rt}}{r}f(t)\,dt + \int_0^{a^{\star}}({a^{\star}}-t)e^{-rt}f(t)\,dt.
\label{densityversion}
\end{equation}

Here is an alternate version, written just in terms of ${}_tp_x$:
\begin{align}
{a^{\star}}&=\int_0^\infty e^{-rt}{}_tp_x\,dt+\int_0^{a^{\star}} \int_s^{a^{\star}}e^{-rt}(1+r({a^{\star}}-t))\,dt f(s)\,ds
\nonumber \\
&=\int_0^\infty e^{-rt}{}_tp_x\,dt+\int_0^{a^{\star}} e^{-rt}(1+r({a^{\star}}-t))(1-{}_tp_x)\,dt 
\nonumber \\
&={a^{\star}}+\int_{a^{\star}}^\infty e^{-rt}{}_tp_x\,dt-\int_0^{a^{\star}} e^{-rt}r({a^{\star}}-t){}_tp_x\,dt 
\label{survivalversion2}
\end{align}
Or in other words,
\begin{equation}
\int_{a^{\star}}^\infty e^{-rt}{}_tp_x\,dt=\int_0^{a^{\star}} e^{-rt}r({a^{\star}}-t){}_tp_x\,dt.
\label{twoaccountversion2}
\end{equation}
This formula appeared earlier as \eqref{twoaccountversion} and, as noted at the time, can be interpreted via making payments first from a phase-one and then from a phase-two account. 

Define functions $F(\alpha)=\int_{\alpha}^\infty e^{-rt}{}_tp_x\,dt$ and $G(\alpha)=\int_0^{\alpha} e^{-rt}r(\alpha-t){}_tp_x\,dt$ of a real variable $\alpha$. So by \eqref{twoaccountversion2} we have that $a^{\star}$ is the value of $\alpha$ making $F(a^{\star})=G({a^{\star}})$. Over the next few pages of this section we plan to examine the dependence of these two expressions or terms on some of the underlying parameters. In terms of notation, we will express these functions as, for example, $F(\alpha;r)$ or $G(\alpha;x)$, but otherwise we will not explicitly show the parameters. 

\begin{theorem} Consider a CRIA with no loading, and let $r>0$. 
\label{CRIAproperties}
\begin{enumerate}
\item There is a unique ${a^{\star}}(x)>0$ at which the CRIA is viable (ie $F({a^{\star}})=G({a^{\star}})$);
\item $r{a^{\star}}(x)<1$
\item ${a^{\star}}(x)$ is $\downarrow$ in $r$ with $\lim_{r\downarrow 0} {a^{\star}}(x)=\infty$, $\lim_{r\to\infty} r{a^{\star}}(x)=1$
\item Assume an increasing hazard rate. Then ${a^{\star}}(x)$ is $\downarrow$ in $x$. If we also assume that $\lim_{x\to\infty}\lambda_x=\infty$ then $\lim_{x\to\infty}{a^{\star}}(x)=0$. 
\end{enumerate}
\end{theorem}

\begin{proof} Note that $F(\alpha)$ is a decreasing function of $\alpha$, starting above 0 when $\alpha=0$ and $\to 0$ when $\alpha\to\infty$. The derivative $\frac{\partial G(\alpha)}{\partial \alpha}=r\int_0^\alpha e^{-rt}{}_tp_x\,dt$ is increasing in $\alpha$, so $G$ is a convex increasing function, starting at 0. Therefore $F$ and $G$ will cross at a unique value of $\alpha$, which therefore defines $a^{\star}$. This shows item (a) in the above theorem. 

\vspace{0.1in}

To prove (b), we use that ${}_tp_x$ is $\downarrow$ in $t$. Therefore at $\alpha={a^{\star}}(x)$ we have that: 
\begin{multline*}
{}_{a^{\star}}p_x \frac{e^{-r{a^{\star}}}}{r} ={}_{a^{\star}}p_x\int_{a^{\star}}^\infty e^{-rt}\,dt>F({a^{\star}})\\
=G({a^{\star}})>{}_{a^{\star}}p_x\int_0^{a^{\star}}e^{-rt}r({a^{\star}}-t)\,dt={}_{a^{\star}}p_x [{a^{\star}} - \frac{1-e^{-r{a^{\star}}}}{r}].
\end{multline*}
This immediately implies that ${a^{\star}}r<1$. Turning to item (c), if $r\downarrow 0$, then $\frac{\partial G}{\partial \alpha}\to 0$ for each $\alpha$, so $G\to 0$ pointwise. That is not true of $F$, so in fact ${a^{\star}}(x)\to\infty$. For $r\to\infty$ we have $ra_x<r{a^{\star}}(x)<1$ by (b). Because ${}_tp_x\to 0$ as $t\to\infty$ it is simple to show that $ra_x\to 1$. Therefore $r{a^{\star}}\to 1$. To show that ${a^{\star}}(x)$ decreases with $r$, observe that $\alpha={a^{\star}}(x)$ is the point where the decreasing function $F(\alpha;r)-G(\alpha;r)$ reaches 0. By (b) of Lemma \ref{CRIApartials} below, increasing $r$ will lower the path $F(\alpha;r)-G(\alpha;r)$ for $\alpha\in[0,\frac1r]$. This interval contains ${a^{\star}}(x)$ by (b), so this forces the zero of $F-G$ to the left. This shows (c). Now assume an increasing hazard rate. The first statement in part (d) follows as above, this time using (c) of Lemma \ref{CRIApartials}. 
Finally, suppose that the hazard rate increases without bound. Fix $\epsilon>0$. Let $\lambda=\lambda_{x+\epsilon}$. Then ${}_tp_x>e^{-\lambda t}$ for $t<\epsilon$ and ${}_tp_x<{}_\epsilon p_xe^{-\lambda(t-\epsilon)}$ for $t>\epsilon$. Therefore 
$$
F(\epsilon)<\int_\epsilon^\infty {}_\epsilon p_x e^{-rt}e^{-\lambda(t-\epsilon)}\,dt=\frac{e^{-r\epsilon}}{r+\lambda}{}_\epsilon p_x=o(\frac{1}{r+\lambda})
$$
since ${}_\epsilon p_x\to 0$ when $x\to\infty$. Likewise 
$$
G(\epsilon)>\int_0^\epsilon e^{-rt}r(\epsilon-t)e^{-\lambda t}\,dt = \frac{r}{r+\lambda}\Big[\epsilon-\frac{1}{r+\lambda}(1-e^{-(r+\lambda)\epsilon})\Big]\sim \frac{r\epsilon}{r+\lambda}
$$
when $x\to\infty$, since also $\lambda\to\infty$. In particular, $F(\epsilon)<G(\epsilon)$ for $x$ sufficiently large, which implies that ${a^{\star}}(x)<\epsilon$. Since $\epsilon$ was arbitrary, in fact ${a^{\star}}(x)\to 0$. 
\end{proof}

The above argument relied on the following result, which we now prove.
\begin{lemma} Consider $F$ and $G$ as above. Then
\label{CRIApartials}
\begin{enumerate}
\item $\frac{\partial F(\alpha)}{\partial \alpha}<0<\frac{\partial G(\alpha)}{\partial \alpha}$.
\item $\frac{\partial F(\alpha;r)}{\partial r}<0<\frac{\partial G(\alpha;r)}{\partial r}$ for $0<\alpha<\frac1r$.
\item Assume an increasing hazard rate. Then $\frac{\partial F(\alpha;x)}{\partial x}<\frac{\partial G(\alpha;x)}{\partial x}$ at $\alpha={a^{\star}}(x)$.
\end{enumerate}
\end{lemma}

\begin{proof}
We've already seen the argument for part (a), at the beginning of the proof of Theorem \ref{CRIAproperties}. For part (b), $\frac{\partial F(\alpha;r)}{\partial r}=-\int_{\alpha}^\infty te^{-rt}{}_tp_x\,dt <0$ and $\frac{\partial G(\alpha;r)}{\partial r}=\int_0^{\alpha}e^{-rt}(\alpha-t){}_tp_x[1-rt]\,dt$. Since $r\alpha<1$, the integrand is $>0$, so $\frac{\partial G(\alpha;r)}{\partial r}>0$.

For part (c) we have 
\begin{align*}
\frac{\partial(F-G)}{\partial x} &=\int_{\alpha}^\infty e^{-rt}{}_tp_x[\lambda_x-\lambda_{x+t}]\,dt-\int_0^{a^{\star}} e^{-rt}r(\alpha-t){}_tp_x[\lambda_x-\lambda_{x+t}]\,dt \\
& =\lambda_xF-\lambda_xG+\int_0^{\alpha} e^{-rt}r(\alpha-t){}_tp_x\lambda_{x+t}\,dt-\int_{\alpha}^\infty e^{-rt}{}_tp_x\lambda_{x+t}\,dt\\
& =\int_0^{a^{\star}} e^{-rt}r({a^{\star}}-t){}_tp_x\lambda_{x+t}\,dt-\int_{a^{\star}}^\infty e^{-rt}{}_tp_x\lambda_{x+t}\,dt
\end{align*}
at $\alpha=a^{\star}$ since $F(a^{\star})=G(a^{\star})$. 
Assume an increasing hazard rate. Then 
$$
\int_0^{a^{\star}} e^{-rt}r({a^{\star}}-t){}_tp_x\lambda_{x+t}\,dt<\int_0^{a^{\star}} e^{-rt}r({a^{\star}}-t){}_tp_x\lambda_{x+{a^{\star}}}\,dt=\lambda_{x+{a^{\star}}}G(a^{\star})
$$
while 
$$
\int_{a^{\star}}^\infty e^{-rt}{}_tp_x\lambda_{x+t}\,dt>\int_{a^{\star}}^\infty e^{-rt}{}_tp_x\lambda_{x+{a^{\star}}}\,dt=\lambda_{x+{a^{\star}}}F(a^{\star}).
$$
Therefore the difference is $<\lambda_{x+{a^{\star}}}(G(a^{\star})-F(a^{\star}))=0$. 
\end{proof}

Note that the statement of item (b) above will be strengthened below (see Lemma \ref{strengthenedpartials}), but we prefer to keep the short proof given above, for the case considered here. Note that with $r=0$ the CRIA is clearly non-viable: We have $G=0$ while $F>0$. For $r<0$ things are even worse -- $F>0$ while $G<0$.  Note also that (c) of the Theorem need not hold, without the assumption of an increasing hazard rate. For example, if the hazard rate decreases, the proof of the Lemma implies that $\frac{\partial [F(\alpha;x)-G(\alpha; x)]}{\partial x}>0$ at $\alpha=a^{\star}$, and therefore $\frac{d{a^{\star}}}{dx}>0$. Though ${a^{\star}}$ declines with $x$, it cannot decline too fast, as the next Lemma shows. Let $y=x+{a^{\star}}(x)$ be the age at which the cash refund expires. Then:
\begin{lemma} 
For a fairly priced CRIA, 
$y$ increases with $x$. 
\label{monotonicityiny}
\end{lemma}

\begin{proof}
Re-express $F$ and $G$ as 
$$
\tilde F(y; x)=\int_{y-x}^\infty e^{-rt}{}_tp_x\,dt, \qquad 
\tilde G(y; x)=r\int_0^{y-x}e^{-rt}(y-x-t){}_tp_x\,dt.
$$
As before, we let
$y=y(x)$ make $\tilde F(y;x)-\tilde G(y;x)=0$. By implicit differentiation, 
$$
\frac{dy}{dx}=-\frac{\frac{\partial [\tilde F-\tilde G]}{\partial x}}{\frac{\partial [\tilde F-\tilde G]}{\partial y}}
$$
at $y=y(x)$. But $\frac{\partial \tilde F}{\partial y}<0$ and $\frac{\partial \tilde G}{\partial y}>0$. And
using integration by parts,
$$
\frac{\partial \tilde F}{\partial x}=e^{-r(y-x)}{}_{y-x}p_x+\int_{y-x}^\infty e^{-rt}{}_tp_x[\lambda_x-\lambda_{x+t}]\,dt=(r+\lambda_x)\tilde F
$$
and 
\begin{align*}
\frac{\partial \tilde G}{\partial x}&=-r\int_0^{y-x} e^{-rt}{}_tp_x\,dt+r\int_0^{y-x}e^{-rt}(y-x-t){}_tp_x[\lambda_x-\lambda_{x+t}]\,dt\\
&=-r\int_0^{y-x} e^{-rt}{}_tp_x\,dt+\lambda_x\tilde G-r(y-x)-r\int_0^{y-x}e^{-rt}[-r(y-x-t)-1]{}_tp_x\,dt\\
&=(r+\lambda_x)\tilde G-r(y-x).
\end{align*}
So the desired inequality holds.
\end{proof}

\subsection{Theorems and proofs: CRIA with loading} 
\label{sec:loadedCRIA}
Next consider viability in the presence of loading. In actuarial terminology, loading is the percentage increase applied to the basic hedging cost to obtain the price charged consumers, with the difference going to fund expenses, capital reserves, or insurer profits. Recall that we assigned the symbol $\pi$ for the proportional loading, and used that in the body of the paper and for the numerical examples in Table \ref{viabler}, with $\pi=5\%,15\%, 25\%$.


\begin{theorem} Consider a CRIA with loading $\pi$. 
\begin{enumerate}
\item If $\frac{\pi}{1+\pi}<ra_x$ then there is a unique ${\widehat{a}^{\star}}(x)$ at which the CRIA is viable. Moreover, ${\widehat{a}^{\star}}(x)$ increases with $\pi$ and decreases with $r$. 
\item If $\frac{\pi}{1+\pi}\ge ra_x$ there is no ${\widehat{a}^{\star}}$ making the CRIA viable. 
\item For each $\pi>0$ and each $x$, there is an $r_{\pi,x}$ such that the CRIA is viable when $r>r_{\pi, x}$ but not when $r\le r_{\pi,x}$. For fixed $x$, $r_{\pi,x}\to 0$ when $\pi\downarrow 0$, and $r_{\pi,x}\to \infty$ when $\pi\uparrow \infty$. If $r\downarrow r_{\pi,x}$ then ${\widehat{a}^{\star}}(x)\to\infty$. If $r\to \infty$ then $r{\widehat{a}^{\star}}(x)\to 1+\pi$. 
\item Assume a hazard rate that increases without bound. For each $\pi>0$ and $r>0$ there is an $x_{\pi,r}$ such that The CRIA is viable when $x<x_{\pi,r}$ but not when $x\ge  x_{\pi,r}$. If $x\uparrow x_{\pi,r}$ then ${\widehat{a}^{\star}}(x)\to\infty$.
\end{enumerate}
\end{theorem}

\begin{proof} Loading changes the LHS of \eqref{survivalversion2} into $\frac{1}{1+\pi}\widehat{a}^{\star}$, and it therefore adds $\frac{\pi}{1+\pi}\widehat{a}^{\star}$ to the LHS of \eqref{twoaccountversion2}. Let $\delta=\frac{\pi}{1+\pi}$. So the CRIA is viable at price ${\widehat{a}^{\star}}$ if $\delta {\widehat{a}^{\star}}+F({\widehat{a}^{\star}})=G({\widehat{a}^{\star}})$. Then $G(\alpha)\le r\alpha\int_0^{\alpha}e^{-rt}{}_tp_x\,dt<r\alpha a_x$ and $\delta \alpha+F(\alpha)>\alpha\delta$. In particular, if $\delta\ge ra_x$ then the curves cannot cross, so no ${\widehat{a}^{\star}}$ can be viable. This shows (b). We know that $F(\alpha)-G(\alpha)$ starts positive at $\alpha=0$, and decreases forever. Its asymptotic slope is $-ra_x$. In particular, if $\delta<ra_x$ then it must cross the curve $-\delta \alpha$, so there is a viable ${\widehat{a}^{\star}}$. The argument also implies that if we let $\delta\uparrow ra_x$, any ${\widehat{a}^{\star}}$ that achieves viability will $\to\infty$. Because $-\delta \alpha$ decreases with $\delta$, we also conclude that the viable ${\widehat{a}^{\star}}$ increases with $\delta$, and hence with $\pi$.  

To see uniqueness of ${\widehat{a}^{\star}}$, observe $F''(\alpha)-G''(\alpha)>0$ so $F-G$ is convex. So if $F(\alpha)-G(\alpha)$ crosses the line $-\delta \alpha$ twice, then after the second crossing it must stay above that line. This is not compatible with the asymptotic slopes, so in fact there is a unique crossing.  To show that ${a^{\star}}(x)$ decreases with $r$, we will use Lemma \ref{strengthenedpartials} below. 

It shows that raising $r$ will lower the curve $F(\alpha;r)-G(\alpha;r)$. This pushes down the value of ${\widehat{a}^{\star}}$ at which this curve crosses the line $-\delta \alpha$, which completes the proof of (a). To address (c) we need to understand how the function $\rho(r)=ra_x$ varies with $r$. We have $\rho(0)=0$, and integration by parts shows that $\rho(r)=1-\int_0^\infty e^{-rt}{}_tp_x\,\lambda_{x+t}\,dt$. This increases with $r$,  converging to 1 as $r\to\infty$. The existence and properties of $r_{\pi,x}$ now follow from our earlier conclusions. To find the asymptotics as $r\to\infty$, observe that the viable ${\widehat{a}^{\star}}$ satisfies $r[F({\widehat{a}^{\star}})-G({\widehat{a}^{\star}})]=-\delta r{\widehat{a}^{\star}}$. Letting $B=r{\widehat{a}^{\star}}$, and changing variables in the integrals, this becomes that
$$
\int_B^\infty e^{-s}{}_{\frac{s}{r}}p_x\,ds-\int_0^B e^{-s}(B-s){}_{\frac{s}{r}}p_x\,ds=-\delta B.
$$
Sending $r\to\infty$ makes this asymptotically $\int_B^\infty e^{-s}\,ds-\int_0^B e^{-s}(B-s)\,ds=-\delta B$, which simplifies to $1-B=-\delta B$. So $B\to \frac{1}{1-\delta}=1+\pi$ which completes the proof of (c). Under the assumptions of (d), if we increase $x$ then $a_x$ decreases, and $a_x\to 0$ when $x\to\infty$. The conclusions of (d) now follow from our earlier statements. 
\end{proof}

The above argument used the following result, which we now prove. It is a strengthening of (b) of Lemma \ref{CRIApartials}.
\begin{lemma}
\label{strengthenedpartials}
Consider $F(\alpha;r)$ and $G(\alpha;r)$ as above. Then $\frac{\partial F}{\partial r}<0<\frac{\partial G}{\partial r}$ for every $\alpha>0$. 
\end{lemma}

\begin{proof}
As in Lemma \ref{CRIApartials},
$\frac{\partial F(\alpha;r)}{\partial r}=-\int_{\alpha}^\infty te^{-rt}{}_tp_x\,dt <0$. Let $b={}_{\frac1r}p_x$. Since ${}_tp_x$ decreases with $t$,  
$$
\frac{\partial G(\alpha;r)}{\partial r}=\int_0^{\alpha}e^{-rt}(\alpha-t){}_tp_x[1-rt]\,dt
>\int_0^{\alpha}be^{-rt}(\alpha-t)[1-rt]\,dt.
$$
The antiderivative required is $-be^{-rt}\Big(\frac1{r^2}+[\frac1r-{\alpha}]t+t^2\Big)$, making the above expression $=\frac{b}{r^2}\Big(1-e^{-r\alpha}(1+\alpha r)\Big)$. 
This is $\ge 0$ since $e^{y}\ge 1+y$. 
\end{proof}

\subsection{Theorems and proofs: IRIA \& Loading} 
\label{sec:IRIA}
Define $H(\alpha)=\frac{1}{r} (1-e^{-r \alpha})+  \int_{\alpha}^{\infty} e^{-rs} \, (_sp_x) \, ds.$ Recalling \eqref{alpha}, an unloaded IRIA is viable at price $a^{\circ}$ if ${a^{\circ}}=H({a^{\circ}})$. More generally, in the presence of loading by $\pi$ , the relevant equation is that 
\begin{equation}
\frac{1}{1+\pi}\widehat{a}^{\circ}=H({\widehat{a}^{\circ}}).
\label{alpha2}
\end{equation}

\begin{theorem}
Consider an IRIA with loading $\pi\ge 0$. Let $r>0$.
\begin{enumerate}
\item There is a unique ${\widehat{a}^{\circ}}(x)$ at which the IRIA is viable. ${\widehat{a}^{\circ}}(x)$ increases with $\pi$.
\item $r{\widehat{a}^{\circ}}(x)<1+\pi$.
\item ${\widehat{a}^{\circ}}(x)$ is $\downarrow$ in $r$. It $\to\infty$ as $r\downarrow 0$, and $r{\widehat{a}^{\circ}}(x)\to 1+\pi$ as $r\to\infty$.
\item Assume an increasing hazard rate. Then ${\widehat{a}^{\circ}}(x)$ is $\downarrow$ in $x$. $\lim_{x\to\infty}{\widehat{a}^{\circ}}(x)>0$ if $\pi >0$. Whereas if the hazard rate is also unbounded and $\pi=0$, then $\lim_{x\to\infty}{a^{\circ}}(x)=0$.
\end{enumerate}
\end{theorem}

\begin{proof}
$H(0)>0$ and $H'(\alpha)=e^{-r\alpha}(1-{}_{\alpha} p_x)<1$. So the viable ${a^{\circ}}$ exists and is unique when $\pi=0$.  
To see this in the case $\pi>0$, rewrite \eqref{alpha2} as
$$
\frac{1}{1+\pi}a^{\circ}-\frac{1-e^{-r{a^{\circ}}}}{r}=\int_{a^{\circ}}^\infty e^{-rs}{}_sp_x\,ds.
$$
Define $u(\alpha)=\frac{1}{1+\pi}\alpha-\frac{1-e^{-r{\alpha}}}{r}$ and $v(\alpha)=\int_{\alpha}^\infty e^{-rs}{}_sp_x\,ds$, so that the above equation becomes that $u(a^{\circ})=v(a^{\circ})$. Then $v(\alpha)$ is $\downarrow$ in $\alpha$ and $>0$. The function $u$ starts at 0, initially decreases, and then increases forever (and $\to +\infty$). Therefore they cross. Moreover, by the time the two paths cross, we're looking at the intersection of an increasing path with a decreasing one, and there can only be one such crossing. This shows (a). 

Now observe that $H(\alpha)<\frac{1-e^{-r\alpha}}{r}+{}_{\alpha} p_x\int_{\alpha}^\infty e^{-rt}\,dt=\frac{1-e^{-\alpha r}[1-{}_{\alpha} p_x]}{r}<\frac1r$. Therefore $\alpha_0=\frac{1+\pi}{r}$ satisfies $H(\alpha_0)<\frac1r=\frac{1}{1+\pi}\alpha_0$. So ${\widehat{a}^{\circ}}(x)<\alpha_0$, showing (b). Turning to (c), $\frac{\partial H(\alpha;r)}{\partial r}=-\int_0^{\alpha} te^{-rt}\,dt-\int_{\alpha}^\infty te^{-rt}{}_tp_x\,dt<0$. So raising $r$ lowers the path $H$ and so decreases its crossing ${\widehat{a}^{\circ}}(x)$ with the path $\frac{1}{1+\pi}\alpha$. When $r\to 0$, we have $u(\alpha)\to-\frac{\pi}{1+\pi}\alpha\le 0$, while $v(\alpha)$ stays bounded from 0 (for any fixed $\alpha$). This forces ${\widehat{a}^{\circ}}(x)\to\infty$. 

When $r\to\infty$ we have $u(\alpha)\to\frac{1}{1+\pi}\alpha$ while $v(\alpha)\to 0$, which forces ${\widehat{a}^{\circ}}(x)\to 0$. More precisely, rewriting $ru({\widehat{a}^{\circ}})=rv({\widehat{a}^{\circ}})$ in terms of $\beta=r{\widehat{a}^{\circ}}$ yields
$$
\frac{1}{1+\pi}\beta-1+e^{-\beta}=\int_\beta^\infty e^{-s}{}_{\frac{s}{r}}p_x\,ds\to e^{-\beta}
$$
as $r\to\infty$. This simplifies to $\frac{1}{1+\pi}\beta\to 1$, which shows (c). 

Assume a rising hazard rate. Observe that function $u(\alpha)$ does not vary with $x$, while the function $v(\alpha)$ declines with $x$ As above, this forces the crossing ${\widehat{a}^{\circ}}(x)$ to decrease as $x$ rises. When $\pi>0$, we have $u(\alpha)\le 0$ for $\alpha$ in some interval $[0,z]$ that does not depend on $x$. Therefore ${\widehat{a}^{\circ}}(x)>z$ for every $x$. But if $\delta=0$ then $u(\alpha)>0$ for $\alpha>0$. If the hazard rate rises without bound then $v(\alpha)\downarrow 0$ when $x\to\infty$, which forces ${a^{\circ}}(x)$ to $\to 0$. 
\end{proof}

Note that the behaviour of a loaded IRIA at advanced ages (ie approaching a non-zero constant) is quite different from that of a CRIA (blowing up) or a LOIA (approaching 0).

\subsection{Sensitivities}
\label{sec:derivatives}
For completeness, we record general formulas for the various derivatives needed to compute sensitivities, which were then evaluated numerically for tables and figures.
\begin{itemize}
\item CRIA $\frac{d{\widehat{a}^{\star}}}{dr}=-\frac{\partial (\frac{\pi}{1+\pi}\alpha+F-G)}{\partial r}/\frac{\partial (\frac{\pi}{1+\pi}\alpha+F-G)}{\partial \alpha}\Big|_{\alpha=\widehat{a}^{\star}}$, where 
$\frac{\partial(F-G)}{\partial \alpha} = -e^{-r\alpha}{}_{\alpha}p_x -r\int_0^{\alpha}e^{-rt}{}_tp_x\,dt$ and 
$\frac{\partial(F-G)}{\partial r} = -\int_{\alpha}^\infty te^{-rt}{}_tp_x\,dt - \int_0^{\alpha}e^{-rt}(\alpha-t){}_tp_x[1-rt]\,dt$. And at $\alpha=\widehat{a}^{\star}$, $F-G=-\frac{\pi}{1+\pi}\alpha$. 
\item CRIA $\frac{d{\widehat{a}^{\star}}}{dx}=-\frac{\partial (\frac{\pi}{1+\pi}\alpha+F-G))}{\partial x}/\frac{\partial (\frac{\pi}{1+\pi}\alpha+F-G)}{\partial \alpha}\Big|_{\alpha=\widehat{a}^{\star}}$, where
$\frac{\partial(F-G)}{\partial x} = \int_0^{\alpha} e^{-rt}r(\alpha-t){}_tp_x\lambda_{x+t}\,dt-\int_{\alpha}^\infty e^{-rt}{}_tp_x\lambda_{x+t}\,dt+\lambda_x[F-G]$, which after applying integration by parts $=r\alpha-e^{-r\alpha}{}_{\alpha}p_x-r\int_0^{\alpha}e^{-rt}{}_tp_x\,dt+(r+\lambda_x)(F-G)$. 
\item LOIA $\frac{da_x}{dr}=-\int_0^\infty te^{-rt}{}_tp_x\,dt$. 
And the loaded LOIA price is $\widehat{a}_x=(1+\pi)a_x$
\item LOIA $\frac{da_x}{dx}=-(1-[r+\lambda_x]a_x)$
\item IRIA $\frac{d{\widehat{a}^{\circ}}}{dr}=\frac{\partial H}{\partial r}/\Big[\frac{1}{1+\pi}-\frac{\partial H}{\partial \alpha}\Big]\Big|_{\alpha=\widehat{a}^{\circ}}$, where 
$\frac{\partial H}{\partial r}=\frac{e^{-r\alpha}(1+\alpha r)-1}{r^2}-\int_{\alpha}^\infty te^{-rt}{}_tp_x\,dt$, and 
$\frac{\partial H}{\partial \alpha}=e^{-r\alpha}(1-{}_{\alpha} p_x)$. And at $\alpha=\widehat{a}^{\circ}$, $H=\frac{1}{1+\pi}\alpha$. 

\item IRIA $\frac{d{\widehat{a}^{\circ}}}{dx}=\frac{\partial H}{\partial x}/\Big[\frac{1}{1+\pi}-\frac{\partial H}{\partial \alpha}\Big]\Big|_{\alpha=\widehat{a}^{\circ}}$, where
$\frac{\partial H}{\partial x}=\int_{\alpha}^\infty e^{-rs}{}_sp_x[\lambda_x-\lambda_{x+s}]\,ds=-e^{-r\alpha}{}_{\alpha} p_x + (r+\lambda_x)\int_{\alpha}^\infty e^{-rs}{}_sp_x\,ds$.

\end{itemize}

\subsection{Theorems and proofs: CRIA dominates IRIA}
\label{sec:inequality}
\begin{theorem} We have ${\widehat{a}^{\star}}> {\widehat{a}^{\circ}}>\widehat{a}_x$ for any $\pi\ge 0$ and $r>0$.
\end{theorem}
\begin{proof}
Recall that $F(\alpha)=\int_{\alpha}^\infty e^{-rt}{}_tp_x\,dt$ represents the insurer's liability due to annuitants who live beyond age $\alpha$. We imagine that the purchase price $\alpha$ is deposited into what we earlier called the phase-one account, which is used for the first $\alpha$ of payments, while interest earned goes into the phase-two account. Loading requires the latter to cover payments to the insurer, whose present value amounts to $\frac{\pi}{1+\pi}\alpha$, in addition to the liability $F(\alpha)$. If $G_i(\alpha)$ represents the present value of interest generated by the payout account (with $i=1,2,3$   corresponding respectively to the LOIA, IRIA, and CRIA versions), then each respective price is determined by matching this interest against liabilities, ie by solving $\frac{\pi}{1+\pi}\alpha + F(\alpha) = G_i(\alpha)$. Since each $G_i(0)=0$ and $F(0)>0$, the desired inequality will follow if we can show that $G_1(\alpha)>G_2(\alpha)>G_3(\alpha)$ for every $\alpha$.  For a CRIA, $G_3(\alpha)=\int_0^{\alpha} e^{-rt}r(\alpha-t){}_tp_x\,dt$ (which was earlier denoted simply $G(\alpha)$). For an IRIA, $G_2(\alpha)=\alpha-\int_0^{\alpha}e^{-rs}\,ds=\alpha-\frac{1-e^{-r\alpha}}{r}$. For a LOIA, $G_1(\alpha)=\alpha-\int_0^{\alpha}e^{-rt}{}_tp_x\,dt$. Because ${}_tp_x<1$, it is clear from the integrals that $G_1(\alpha)<G_2(\alpha)$. For the same reason, $G_3(\alpha)<\int_0^{\alpha} e^{-rt}r(\alpha-t)\,dt$, which is easily seen to $=G_2(\alpha)$. 
\end{proof}

\subsection{Macaulay Duration}
\label{sec:Macaulay}

As described earlier, our definition of duration differs from traditional Macaulay duration in the case of CRIAs and IRIAs. But for completeness, we give the Macaulay versions here. If cash flows at rate $\phi_t$ at time $t$, the associated Macaulay duration is $\frac{\int_0^\infty t e^{-rt}\phi_t\,dt}{\int_0^\infty e^{-rt}\phi_t\,dt}$, measured in units of time. In other words the arrival times $t$ are weighted by the present values $e^{-rt}\phi_t$ of the cash flows. So if $M(r)=\int_0^\infty e^{-rt}\phi_t\,dt$, life annuity duration is $-\frac{1}{M}\frac{\partial M}{\partial r}$. Note that the cash flows $\phi_t$ are being held constant here, rather than adjusting with $r$. In the case of a CRIA, the present value of cash flows is $M(r,{\widehat{a}^{\star}})={\widehat{a}^{\star}}+F({\widehat{a}^{\star}},r)-G({\widehat{a}^{\star}},r)$ where ${\widehat{a}^{\star}}$ is obtained by solving $\frac{1}{1+\pi}\widehat{a}^{\star}=M(r,{\widehat{a}^{\star}})$. Therefore the formula for Macaulay duration, as defined above, now becomes $-\frac{1+\pi}{{\widehat{a}^{\star}}}\frac{\partial M}{\partial r}$. This is now easily computed. We can understand the asymptotics as $x$ increases to the maximum feasible age. As that happens, we know ${\widehat{a}^{\star}}\to\infty$. Based on the discussion in Section \ref{sec:Analysis} of the paper, life annuity duration converges to $(1+\pi)\int_0^\infty e^{-rt}{}_tp_x(1-rt)\,dt>0$. To see the inequality, observe that the sign of the integrand and the $\downarrow$ nature of ${}_tp_x$ makes this $>\int_0^\infty e^{-rt}{}_{\frac1r}p_x(1-rt)\,dt$, which $=0$ by calculus.

\subsection{Life Annuity Duration} 
\label{sec:annuityduration}
In the case of a bond or a LOIA, Macaulay duration agrees with {\em life annuity duration} $Dr[a]$ as defined in Section \ref{sec:Macaulay}. As above, for a CRIA, the actuarial present value of payments is $M(r,{\widehat{a}^{\star}})={\widehat{a}^{\star}}+F({\widehat{a}^{\star}},r)-G({\widehat{a}^{\star}},r)$.
Therefore with $Dr[\widehat{a}^{\star}]=-\frac{1}{{\widehat{a}^{\star}}}\frac{d{\widehat{a}^{\star}}}{dr}$, and Macaulay duration as defined as above, it follows from the formulas of Section \ref{sec:derivatives} that  
$$
\text{Life Annuity Duration} = \frac{\text{Macaulay Duration}}{(1+\pi)\frac{\partial(\frac{\pi}{1+\pi} {\widehat{a}^{\star}}+F-G)}{\partial {\widehat{a}^{\star}}}}.
$$
One can easily see the relationship between the two measure and they are clearly not equal. Consider now the asymptotics of $Dr[\widehat{a}^{\star}]$ as $x$ rises to the maximal feasible age. ${\widehat{a}^{\star}}$ is the value of $\alpha$ at which the convex function $F(\alpha)-G(\alpha)$ equals the function $-\frac{\pi}{1+\pi} \alpha$. If that crossing takes place at very large $\alpha$, that forces the graphs of the two functions to run nearly parallel. At the limiting value of $x$, the convex curve will in fact be asymptotic to the line. In other words, as $x$ approach this value, $\lim{\alpha\to\infty}\frac{\partial(F-G)}{\partial \alpha}\uparrow -\frac{\pi}{1+\pi}$. Therefore the denominator in the fraction displayed above will $\downarrow 0$, so $Dr[\widehat{a}^{\star}]\to \infty$. That is, $\frac{d{\widehat{a}^{\star}}}{dr}$ blows up even faster than ${\widehat{a}^{\star}}$ does. 

\subsection{Asymptotics for exponential mortality}
\label{sec:exponential}
An example where explicit calculations are possible is that of exponential mortality, at rate $\lambda$. Here are the asymptotics of prices and {\em life annuity durations} as $r\downarrow 0$, in the unloaded case. For CRIA and IRIA we give two terms in the expansion. 
\begin{itemize}
\item{LOIA:} price $=\frac{1}{r+\lambda}\to \frac{1}{\lambda}$ and life annuity duration $=\frac{1}{r+\lambda}\to \frac{1}{\lambda}$.
\item{CRIA:} price $\approx\frac{1}{\lambda}\Big[\log(\frac1r)-\log\log(\frac1r)\Big]$ and life annuity duration $\approx\frac{1}{r\log(\frac1r)}\Big[1+\frac{\log\log(\frac1r)}{\log(\frac1r)}\Big]$
\item{IRIA:} price $\approx\frac{1}{\lambda}\Big[\log(\frac1r)-2\log\log(\frac1r)\Big]$ and life annuity duration $\approx\frac{1}{r\log(\frac1r)}\Big[1+2\frac{\log\log(\frac1r)}{\log(\frac1r)}\Big]$
\end{itemize}
We see that the LOIA price stays bounded while the IRIA and CRIA prices blow up, but with the CRIA price higher. But the IRIA duration is the highest.

Here are the asymptotics of prices and {\em life annuity durations} as $r\to\infty$. Again, we give two terms in the expansion. 
\begin{itemize}
\item{LOIA:} price $=\frac{1}{r+\lambda}\approx \frac{1}{r}\Big[1-\frac{\lambda}{r}\Big]$ and life annuity duration $=\frac{1}{r+\lambda}\approx \frac{1}{r}\Big[1-\frac{\lambda}{r}\Big]$.
\item{CRIA:} price $\approx\frac{1}{r}\Big[1-\frac{\lambda}{r}(1-\frac1e)\Big]$ and life annuity duration $\approx\frac{1}{r}\Big[1-\frac{\lambda}{r}(1-\frac1e)\Big]$
\item{IRIA:} price $\approx\frac{1}{r}\Big[1-\frac{\lambda}{r}(\frac2e)\Big]$ and life annuity duration $\approx\frac{1}{r}\Big[1-\frac{\lambda}{r}(\frac2e)\Big]$
\end{itemize}
Since $1>\frac2e>1-\frac1e$, here the CRIA values $>$ the IRIA values $>$ the LOIA values. 

\section{Conclusion}
\label{sec:conclusions}

Although income annuities are the least popular of all annuities sold in the U.S., comprising less than $5\%$ of all annuities sold in the U.S. during the first quarter of 2021 according to a survey by LIMRA, they continue to be the focus of much academic and scholarly interest. And, as noted in the literature review (Section \ref{sec:literature}), most of the research on income annuities tends to focus on contracts with life only features, or those with fixed guarantee periods, whereas the majority of sales include a cash refund and instalment refund feature. Moreover, as noted in the introduction, in late September 2021 approximately 77 million Americans with a defined contribution plan will come face-to-face with annuity illustrations on their statements.

\vspace{0.1in} 

This is more than just an empirical observation about the type of insurance riders that are preferred by consumers. Although motivated by these empirical choices the main message of this paper is a theoretical one, namely that cash refund income annuities lead to some very interesting and complex valuation problems. As proved in Section \ref{sec:theorems}, if valuation rates sink below a certain threshold the CRIA is no longer viable and can no longer be offered. This might help explain why inflation-linked income annuities no longer exist in the U.S. market place; the term structure of real rates have likely declined under that critical threshold.

\vspace{0.1in}

There are a number of possible avenues for further research that follow from the insights provided by this introductory article. First, it would be very interesting to dig deeper into the money’s worth ratio (MWR) of cash-refund income annuities (IA) to see whether their values are consistently higher that the MWRs of life-only and period certain IAs. It’s quite possible CRIA reduces the extent of anti-selection and that MWR values are higher for the consumer. Table  \ref{MWRtable} provides preliminary evidence of that.

\vspace{0.1in}

Second, and as a follow-up, perhaps CRIAs and their embedded life insurance provide an ``optimized” way of fulfilling a bequest motive, which might help explain the relative demand for cash refunds over life only riders. Exploring how anti-selection and mortality heterogeneity affects the valuation and pricing of cash refund income annuities, is yet another avenue worth pursuing. While there is a large body of literature focused on optimal life-only annuitization strategies, it would be interesting to explore how those result change once a cash-refund product is introduced and competes with a life only version.

\vspace{0.1in}

Third and finally, this entire paper has been predicated on the {\em law of large numbers} that assumes the insurance company issuing the IA is selling enough so that mortality risk is entirely diversifiable and pricing is done by (risk neutral) expectations. In reality of course annuitant pools are of a finite size, especially as the overall demand for income annuities is rather thin. Just as importantly, mortality rates are {\em not} deterministic and Covid-19 was just one data point that supports a {\em stochastic mortality} perspective. It would be interesting to measure the extent to which {\em ruin probabilities} and the corresponding amount of capital required is higher for CRIA. In time-honoured tradition we leave these for subsequent research.

\clearpage
\bibliographystyle{jf}
\bibliography{CRIAbib}

\begin{thebibliography}{53}
\expandafter\ifx\csname natexlab\endcsname\relax\def\natexlab#1{#1}\fi

\bibitem[Alexandrova and Gatzert(2019)]{AG2019}
Alexandrova, M., and N.~Gatzert, 2019, What do we know about annuitization
  decisions?, {\em Risk Management and Insurance Review\/} 22, 57--100.

\bibitem[Bauer et~al.(2008)Bauer, Kling, and Russ]{BKR2008}
Bauer, D., A.~Kling, and J.~Russ, 2008, A universal pricing framework for
  guaranteed minimum benefits in variable annuities,, {\em ASTIN Bulletin: The
  Journal of the IAA\/} 38, 621--651.

\bibitem[Benartzi et~al.(2011)Benartzi, Previtero, and Thaler]{BPT2011}
Benartzi, S., A.~Previtero, and R.~H. Thaler, 2011, Annuitization puzzles, {\em
  Journal of Economic Perspectives\/} 25, 143--164.

\bibitem[Blanchett et~al.(2021)Blanchett, Finke, and Nikolic]{BFN2021}
Blanchett, D., M.~Finke, and B.~Nikolic, 2021, How competitive are income
  annuity providers over time?, {\em Risk Management and Insurance Review\/}
  24, 207--214.

\bibitem[Bodie(1990)]{B1990}
Bodie, Z., 1990, Pensions as retirement income insurance, {\em Journal of
  Economic Literature\/} 28, 28--49.

\bibitem[Bommier and Le~Grand(2014)]{BL2014}
Bommier, A., and F.~Le~Grand, 2014, Too risk averse to purchase insurance?,
  {\em Journal of Risk and Uncertainty\/} 48, 135--166.

\bibitem[Bowers et~al.(1997)Bowers, Gerber, Hickman, Jones, and
  Nesbitt]{BGHJN1997}
Bowers, N.~L., H.~U. Gerber, J.~C. Hickman, D.~A. Jones, and C.J. Nesbitt,
  1997, {\em Actuarial Mathematics\/} (Society of Actuaries, Schaumburg, IL.).

\bibitem[Boyer et~al.(2020)Boyer, Box-Couillard, and Michaud]{BBM2020}
Boyer, M.~M., S.~Box-Couillard, and P.~C. Michaud, 2020, Demand for annuities:
  Price sensitivity, risk perceptions, and knowledge, {\em Journal of Economic
  Behavior and Organization\/} 180, 883--902.

\bibitem[Brown et~al.(2001)Brown, Mitchell, Poterba, and Warshawsky]{BMPW2001}
Brown, J.~R., O.S. Mitchell, J.M. Poterba, and M.J. Warshawsky, 2001, {\em The
  Role of Annuity Markets in Financing Retirement\/} (MIT Press, Cambridge).

\bibitem[Brown et~al.(2008)Brown, Kling, Mullainathan, and Wrobel]{BKMW2008}
Brown, J.R., J.R. Kling, S.~Mullainathan, and M.V. Wrobel, 2008, Why don't
  people insure late-life consumption? a framing explanation of the
  under-annuitization puzzle, {\em American Economic Review\/} 98, 304--309.

\bibitem[Cannon and Tonks(2008)]{CT2008}
Cannon, E., and I.~Tonks, 2008, {\em Annuity Markets\/} (Oxford University
  Press, New York).

\bibitem[Chai et~al.(2011)Chai, Horneff, Maurer, and Mitchell]{CHMM2011}
Chai, J., W.~Horneff, R.~Maurer, and O.~S. Mitchell, 2011, Optimal portfolio
  choice over the life cycle with flexible work, endogenous retirement, and
  lifetime payouts, {\em Review of Finance\/} 15, 875--907.

\bibitem[Charupat et~al.(2016)Charupat, Kamstra, and Milevsky]{CKM2016}
Charupat, N., M.~J. Kamstra, and M.~A. Milevsky, 2016, The sluggish and
  asymmetric reaction of life annuity prices to changes in interest rates, {\em
  Journal of Risk and Insurance\/} 83, 519--555.

\bibitem[Davidoff et~al.(2005)Davidoff, Brown, and Diamond]{DBD2005}
Davidoff, T., J.H. Brown, and P.A. Diamond, 2005, Annuities and individual
  welfare, {\em American Economic Review\/} 95, 1573--1590.

\bibitem[Davies(1981)]{D1981}
Davies, J.B., 1981, Uncertain lifetime, consumption and dissaving in
  retirement, {\em Journal of Political Economy\/} 89, 561--577.

\bibitem[Dickson et~al.(2009)Dickson, Hardy, and Waters]{DHW2009}
Dickson, D.C.M., M.R. Hardy, and H.~R. Waters, 2009, {\em Actuarial mathematics
  for life contingent risks\/} (Cambridge University Press, U.K.).

\bibitem[Finkelstein and Poterba(2004)]{FP2004}
Finkelstein, A., and J.M. Poterba, 2004, Adverse selection in insurance
  markets: Policyholder evidence from the u.k. annuity market, {\em Journal of
  Political Economy\/} 112, 183--208.

\bibitem[Finkelstein and Poterba(2014)]{FP2014}
Finkelstein, A., and J.M. Poterba, 2014, Testing for asymmetric information
  using unused observables in insurance markets: Evidence from the uk annuity
  market, {\em Journal of Risk and Insurance\/} 81, 709--734.

\bibitem[Friedman and Warshawsky(1990)]{FW1990}
Friedman, B.~M., and M.~J. Warshawsky, 1990, The cost of annuities:
  Implications for saving behavior and bequests, {\em The Quarterly Journal of
  Economics\/} 105, 135--154.

\bibitem[Gompertz(1825)]{G1825}
Gompertz, B., 1825, On the nature of the function expressive of the law of
  human mortality and on a new mode of determining the value of life
  contingencies, {\em Philosophical Transactions of the Royal Society of
  London\/} 115, 513--583.

\bibitem[Gong and Webb(2010)]{GW2010}
Gong, G., and A.~Webb, 2010, Evaluating the advanced life deferred annuity: An
  annuity people might actually buy, {\em Insurance: Mathematics and
  Economics\/} 46, 210--221.

\bibitem[Haberman and Sibbett(1995)]{HS1995}
Haberman, S., and T.A. Sibbett, eds., 1995, {\em History of Actuarial Science
  (10 volumes)\/} (William Pickering, London).

\bibitem[Habib et~al.(2020)Habib, Huang, Mauskopf, Nikolic, and
  Salisbury]{HHMNS2020}
Habib, F., H.~Huang, A.~Mauskopf, B.~Nikolic, and T.S. Salisbury, 2020, Optimal
  allocation to deferred income annuities, {\em Insurance: Mathematics and
  Economics\/} 90, 94--104.

\bibitem[Halley(1693)]{H1693}
Halley, E., 1693, An estimate of the degrees of the mortality of mankind, drawn
  from the curious tables of the births and funerals at the city of breslaw,
  {\em Philosophical Transactions of the Royal Society of London\/} 17,
  596--610.

\bibitem[Horneff et~al.(2020)Horneff, Maurer, and Mitchell]{HMM2020}
Horneff, V., R.~Maurer, and O.~S. Mitchell, 2020, Putting the pension back in
  401 (k) retirement plans: Optimal versus default deferred longevity income
  annuities, {\em Journal of Banking and Finance (in press)\/} 414.

\bibitem[Horneff et~al.(2009)Horneff, Maurer, Mitchell, and Stamos]{HMMS2009}
Horneff, W.~J., R.~H. Maurer, O.~S. Mitchell, and M.~Z. Stamos, 2009, Asset
  allocation and location over the life cycle with investment-linked
  survival-contingent payouts, {\em Journal of Banking and Finance\/} 33,
  1688--1699.

\bibitem[Huang et~al.(2017{\natexlab{a}})Huang, Milevsky, and Young]{HMY2017}
Huang, H., M.A. Milevsky, and V.R. Young, 2017{\natexlab{a}}, Optimal
  purchasing of deferred income annuities when payout yields are
  mean-reverting, {\em Review of Finance\/} 21, 327--361.

\bibitem[Huang et~al.(2017{\natexlab{b}})Huang, Zeng, and Kwok]{HZK2017}
Huang, Y.~T., P.~Zeng, and Y.~K. Kwok, 2017{\natexlab{b}}, Optimal initiation
  of guaranteed lifelong withdrawal benefit with dynamic withdrawals, {\em SIAM
  Journal on Financial Mathematics\/} 8, 804--840.

\bibitem[Inkmann et~al.(2011)Inkmann, Lopes, and Michaelides]{ILM2011}
Inkmann, J., P.~Lopes, and A.~Michaelides, 2011, How deep is the annuity market
  participation puzzle?, {\em The Review of Financial Studies\/} 24, 279--319.

\bibitem[Kingston and Thorp(2005)]{KT2005}
Kingston, G., and S.~J. Thorp, 2005, Annuitization and asset allocation with
  hara utility, {\em Journal of Pension Economics and Finance\/} 4, 225--248.

\bibitem[Knoller(2016)]{K2016}
Knoller, C., 2016, Multiple reference points and the demand for
  principal-protected life annuities: An experimental analysis, {\em Journal of
  Risk and Insurance\/} 83, 163--179.

\bibitem[Koijen et~al.(2011)Koijen, Nijman, and Werker]{KNW2011}
Koijen, R.~S., T.~E. Nijman, and B.~J. Werker, 2011, Optimal annuity risk
  management, {\em Review of Finance\/} 15, 799--833.

\bibitem[Merton(2014)]{M2014}
Merton, R.~C., 2014, The crisis in retirement planning, {\em Harvard Business
  Review\/} 92, 43--50.

\bibitem[Milevsky(2006)]{M2006}
Milevsky, M.~A., 2006, {\em The Calculus of Retirement Income: Financial Models
  for Pension Annuities and Life Insurance\/} (Cambridge University Press, New
  York).

\bibitem[Milevsky and Salisbury(2006)]{MS2006}
Milevsky, M.~A., and T.~S. Salisbury, 2006, Financial valuation of guaranteed
  minimum withdrawal benefits, {\em Insurance: Mathematics and Economics\/} 38,
  21--38.

\bibitem[Milevsky and Young(2007)]{MY2007}
Milevsky, M.~A., and V.~R. Young, 2007, Annuitization and asset allocation,
  {\em Journal of Economic Dynamics and Control\/} 31, 3138--3177.

\bibitem[Mitchell et~al.(1999)Mitchell, Poterba, Warshawsky, and
  Brown]{MPWB1999}
Mitchell, O.~S., J.~M. Poterba, M.~J. Warshawsky, and J.~R. Brown, 1999, New
  evidence on the money's worth of individual annuities, {\em American Economic
  Review\/} 89, 1299--1318.

\bibitem[Modigliani(1988)]{M1988}
Modigliani, F., 1988, The role of intergenerational transfers and life cycle
  saving in the accumulation of wealth, {\em Journal of Economic
  Perspectives\/} 2, 15--40.

\bibitem[Moenig and Bauer(2016)]{MB2016}
Moenig, T., and D.~Bauer, 2016, Revisiting the risk-neutral approach to optimal
  policyholder behavior: A study of withdrawal guarantees in variable
  annuities, {\em Review of Finance\/} 20, 759--794.

\bibitem[Ngai and Sherris(2011)]{NS2011}
Ngai, A., and M.~Sherris, 2011, Longevity risk management for life and variable
  annuities: The effectiveness of static hedging using longevity bonds and
  derivatives, {\em Insurance: Mathematics and Economics\/} 49, 100--114.

\bibitem[Pashchenko(2013)]{P2013}
Pashchenko, S., 2013, Accounting for non-annuitization, {\em Journal of Public
  Economics\/} 98, 53--67.

\bibitem[Pitacco(2016)]{P2016}
Pitacco, E., 2016, Guarantee structures in life annuities: A comparative
  analysis, {\em The Geneva Papers on Risk and Insurance-Issues and Practice\/}
  41, 78--97.

\bibitem[Poterba(1997)]{P1997}
Poterba, J.~M., 1997, The history of annuities in the united states, {\em
  National Bureau of Economic Research (NBER), Working Paper\/} 6001.

\bibitem[Poterba and Solomon(2021)]{PS2021}
Poterba, J.~M., and A.~Solomon, 2021, Discount rates, mortality projections,
  and money's worth calculations for us individual annuities, {\em National
  Bureau of Economic Research (NBER), Working Paper\/} 28557.

\bibitem[Poterba(2014)]{P2014}
Poterba, J.M., 2014, Retirement security in an aging population, {\em American
  Economic Review\/} 104, 1--30.

\bibitem[Promislow(2006)]{P2006}
Promislow, S.~D., 2006, {\em Fundamentals of actuarial mathematics\/} (John
  Wiley \& Sons, Toronto).

\bibitem[Reichling and Smetters(2015)]{RS2015}
Reichling, F., and K.~Smetters, 2015, Optimal annuitization with stochastic
  mortality and correlated medical costs, {\em American Economic Review\/} 105,
  3273--3320.

\bibitem[Sheshinski(2008)]{S2008}
Sheshinski, E., 2008, {\em The Economic Theory of Annuities\/} (Princeton
  University Press, Princeton).

\bibitem[Sheshinski(2010)]{S2010}
Sheshinski, E., 2010, Refundable annuities (annuity options), {\em Journal of
  Public Economic Theory\/} 12, 7--21.

\bibitem[Steinorth and Mitchell(2015)]{SM2015}
Steinorth, P., and O.~S. Mitchell, 2015, Valuing variable annuities with
  guaranteed minimum lifetime withdrawal benefits, {\em Insurance: Mathematics
  and Economics\/} 64, 246--258.

\bibitem[Wettstein et~al.(2021)Wettstein, Munnell, Hou, and Gok]{WMHG2021}
Wettstein, G., Al. Munnell, W.~Hou, and N.~Gok, 2021, The value of annuities,
  {\em SSRN: https://ssrn.com/abstract=3797822\/} .

\bibitem[Xu et~al.(2018)Xu, Chen, Coleman, and Coleman]{XCCC2018}
Xu, W., Y.~Chen, C.~Coleman, and T.~F. Coleman, 2018, Moment matching machine
  learning methods for risk management of large variable annuity portfolios,
  {\em Journal of Economic Dynamics and Control\/} 87.

\bibitem[Yaari(1965)]{Y1965}
Yaari, M.~E., 1965, Uncertain lifetime, life insurance, and the theory of the
  consumer, {\em The Review of Economic Studies\/} 32, 137--150.

\end{thebibliography}

\newpage

\section{Appendix: R-script for computing CRIA values}
\label{sec:appendix}

The following R-script uses the bisection method (with 100 iterations) together with R's built-in {\tt integrate} function to compute the value of a cash-refund income annuity (CRIA) at age $x$, under a Gompertz law of mortality with parameters $(m,b)$ and discount rate $r$. The integration scheme is based on the representation provided in equation \eqref{twoaccountversion} in the body of the paper. In Section \ref{sec:theorems} we used functions $F$ and $G$ to represent the left and right-hand side. The resulting annuity price ${a^{\star}}$ would be accurate to within \$1 per million premium.
\begin{Verbatim}[frame=single]
CRIA<-function(x,r,m,b){
a_left<-0; a_right<-1/r
for (i in 1:100){
a<-(a_left+a_right)/2
FINT<-function(t){exp(-r*t)*exp(exp((x-m)/b)*(1-exp(t/b)))}
GINT<-function(t){exp(-r*t)*r*(a-t)*exp(exp((x-m)/b)*(1-exp(t/b)))}
f<-integrate(FINT,a,Inf)$value-integrate(GINT,0,a)$value
if (abs(f)<0.000001) {break}
if (f<0){a_right<-a} else {a_left<-a}
} a }
\end{Verbatim}
For example, assuming $(x=65,m=90,b=10)$, the annual income generated by a premium of $1,\!000,\!000$ in a CRIA, with discount rates ranging from $r=1\%$ to $r=4\%$ would range from \$51,164 to \$67,103 per year (paid in continuous time) and computed as:
\begin{Verbatim}[frame=single]
> 1000000/CRIA(65,0.02,90,10)
[1] 51164.71
> 1000000/CRIA(65,0.03,90,10)
[1] 59169.89
> 1000000/CRIA(65,0.04,90,10)
[1] 67103.97
\end{Verbatim}
For comparison, the life-only version at $r=2\%$, would generate \$7,508 more income.
\begin{Verbatim}[frame=single]
> x<-65; m<-90; b<-10; r<-0.02
> FINT<-function(t){exp(-r*t)*exp(exp((x-m)/b)*(1-exp(t/b)))}
> a<-integrate(FINT,0,Inf)$value
> 1000000/a
[1] 58672.44
\end{Verbatim}

\clearpage


\begin{table}[t]
\begin{center}
\begin{tabular}{||c||c||c||}
\hline\hline
\textbf{Type of income Annuity} & \textbf{Q1.2021} & \textbf{Q4.2011} \\ \hline\hline
Life Only (with no guarantee) & 10.6\% & 25.3\% \\ \hline
Life with Period Certain & 30.0\% & 56.2\% \\ \hline
Refundable: Cash or Instalment & 59.4\% & 18.5\% \\ \hline
{\bf TOTAL:} & {\bf 100\%} & {\bf 100\%} \\ \hline
\end{tabular}
\caption{Market quotes for Income Annuities (IA) in the U.S., compiled directly by the authors based on aggregate data provided by CANNEX Financial Exchanges from their quarterly distributor activity survey experience. Note over the last decade the sharp increase in demand for {\em cash refund} and {\em instalment refund} IAs, and corresponding decline in {\em life only} IAs.}
\label{salesdata}
\end{center}
\vspace{-0.2in}
\end{table}

\clearpage

\begin{table}[h]
\begin{center}
\begin{tabular}{||c||c|c||}
\hline \hline
Age $(x)$ & $r=2\%$ & $r=4\%$ \\ \hline
 55 & 22.12615  & 16.82003   \\ \hline
 65 & 17.04378  & 13.73359   \\ \hline
 75 & 11.91615 & 10.17229 \\ \hline\hline
\end{tabular}
\medskip
\caption{Life Only Income Annuity (LOIA) prices $a$, with Gompertz mortality ($m=90, b=10$). For example under a $2\%$ valuation rate with no loadings, a \$100,000 premium paid at the age of 65 would generate an income of $\$5,867.24=\frac{100000}{17.04378}$ per year, but would terminate upon death of the annuitant.}
\label{LOIA_numbers}
\end{center}
\end{table}

\begin{table}[h]
\begin{center}
\begin{tabular}{||c||c|c||}
\hline \hline
Age $(x)$ & $r=2\%$ & $r=4\%$ \\ \hline
 55 & 23.79569  & 17.47113   \\ \hline
 65 & 19.54472  & 14.90225   \\ \hline
 75 & 15.19471 & 11.97156 \\ \hline\hline
\end{tabular}
\medskip
\caption{Cash Refund Income Annuity (CRIA) prices, with Gompertz mortality $(m=90,b=10)$. For example under a $2\%$ valuation rate with no loadings, a \$100,000 premium at age 65 would generate an income of $\$5,116.47=\frac{100000}{19.54472}$ per year, but upon death of the annuitant the beneficiary would receive: $\$100,000$ minus the cumulative income received, if positive.}
\label{CRIA_numbers}
\end{center}
\end{table}

\begin{table}[h]
\begin{center}
\begin{tabular}{||c||c|c||}
\hline \hline
Age $(x)$ & $r=2\%$ & $r=4\%$ \\ \hline
 55 & 23.55514  & 17.34376   \\ \hline
 65 & 19.18235  & 14.68173   \\ \hline
 75 & 14.7048 & 11.63911 \\ \hline\hline
\end{tabular}
\medskip
\caption{Instalment Refund Income Annuity (IRIA) prices, with Gompertz mortality ($m=90,b=10$). For example under a $2\%$ valuation rate with no loadings a \$100,000 premium at age 65 would generate $\$5,213.13=\frac{100000}{19.18235}$ per year, which is slightly more than the CRIA. Upon death of the annuitant the beneficiary would continue to receive $\$5,213$ until the $\$100,000$ premium was returned, instead of a lump-sum.}
\label{IRIA_numbers}
\end{center}
\end{table}

\clearpage

\begin{table}[h]
\begin{center}
\begin{tabular}{||c||c|c|c||}
\hline\hline
Age & \multicolumn{3}{c||}{Insurance Loading:}\\ \hline
 $(x)$ & $\pi=5\%$ \,& $\pi=15\%$ & $\pi=25\%$ \\ \hline\hline
 55 &16  & 46  & 74  \\ \hline
 65 & 23  & 65  & 105 \\ \hline
 75 & 35 & 101 & 163 \\ \hline\hline
\end{tabular}
\medskip
\caption{The lowest viable valuation rate denoted by $r_{\pi, x}$, expressed in basis points, under which the Cash Refund Income Annuity (CRIA) is still feasible with Gompertz Mortality ($m=90$ and $b=10$). Thus, for example, an inflation-adjusted CRIA would not be feasible at $x=75$, under a loading of $\pi=15\%$, if real rates used for pricing fall under $r_{0.15, 75}=1.01\%$.}
\label{viabler}
\end{center}
\end{table}

\begin{table}[h]
\begin{center}
\begin{tabular}{||c||c|c||c|c||}
\hline\hline
{\bf Age \&} & \multicolumn{4}{c||}{{\bf Real \& Live Quotes for \$100,000 in IA Premium}} \\ \hline 
{\bf Gender} & {\bf Life Only IA} & $\rightarrow$ {\em MWR} & {\bf Cash Refund IA} & $\rightarrow$ {\em MWR} \\ \hline 
65 M & \$5,844 & 0.996 & \$5,280 & 1.031 \\ \hline
65 F & \$5,556 & 1.005 & \$5,112 & 1.043 \\ \hline
80 M & \$10,524 & 1.002 & \$7,788 & 1.017 \\ \hline
80 F & \$9,636 & 1.008 & \$7,428 & 1.033 \\ \hline
\end{tabular}
\medskip
\caption{The {\em Money's Worth Ratio} is computed using an $r=2\%$ valuation rate, under Gompertz mortality with $m=90$ for males, $m=92$ for females and $b=10$ years for both. Data source is the website of {\em Fidelity Investments} on 10 July 2021, using non-qualified funds. Note {\em MWR} is higher for the cash-refund IA when the same mortality assumptions are used.}
\label{MWRtable}
\end{center}
\end{table}


\clearpage

\begin{figure}
\begin{center}
\includegraphics[width=0.45\textwidth]{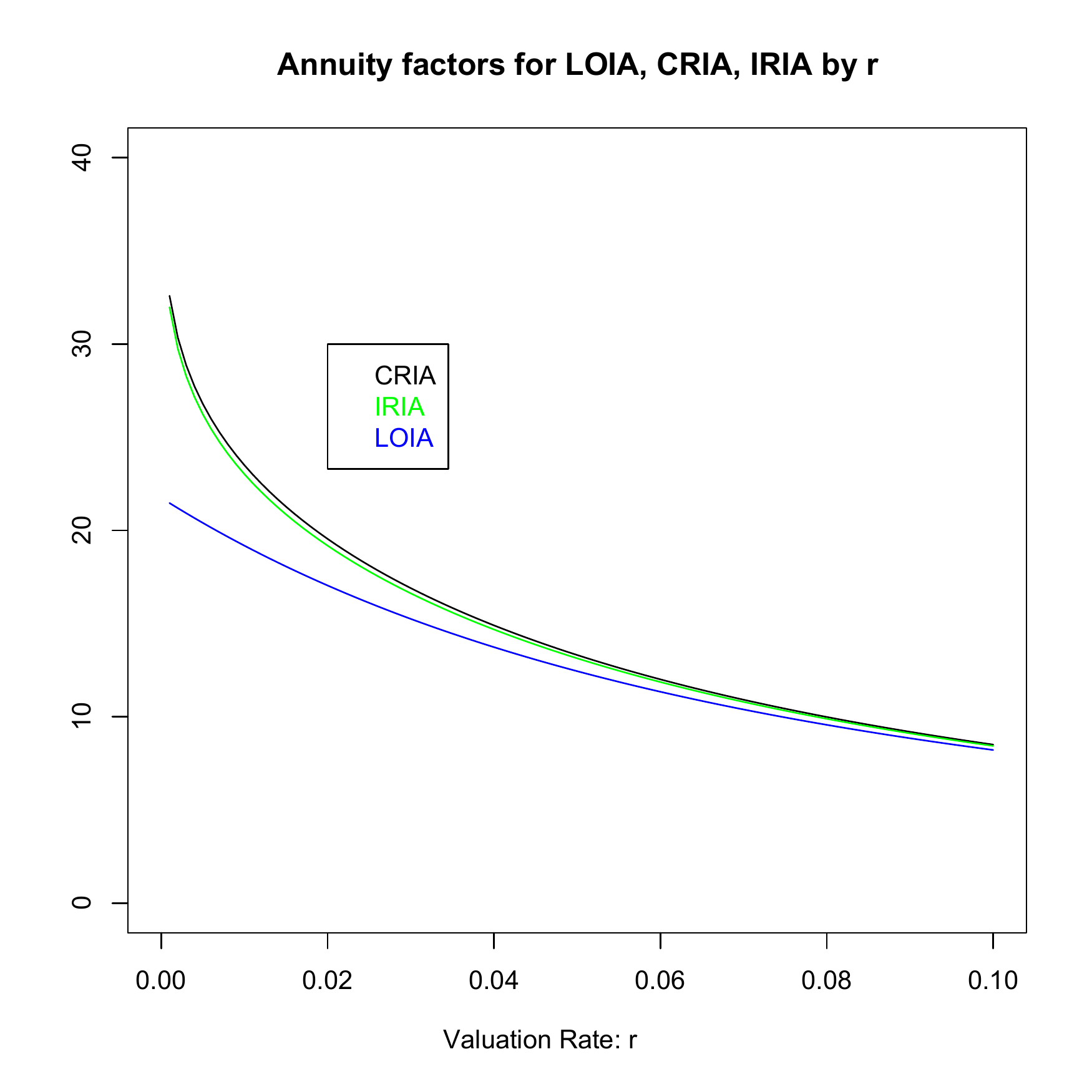} 
\includegraphics[width=0.45\textwidth]{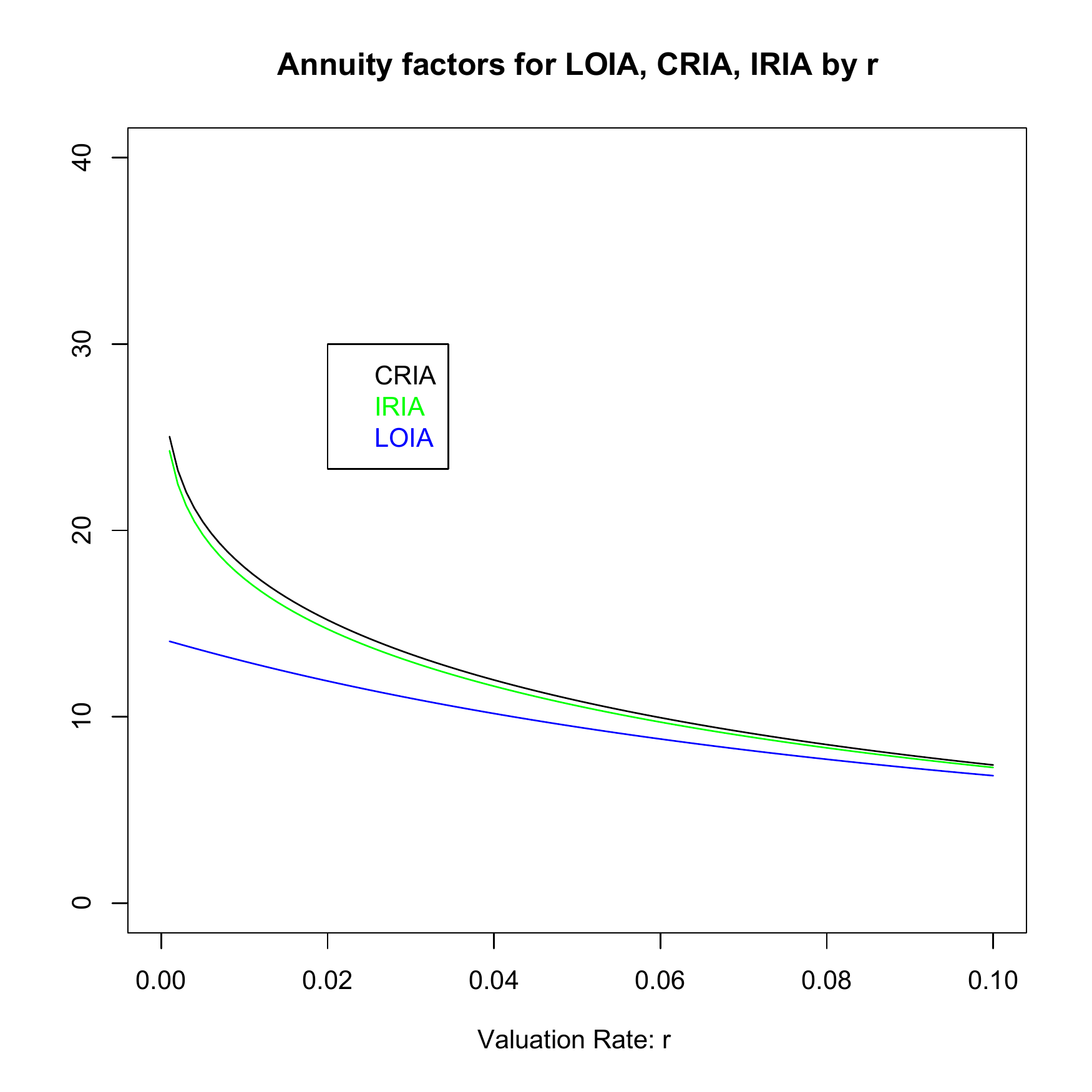} 
\caption{Income Annuity (IA) prices (unloaded) under a variety of valuation rates, assuming Gompertz mortality $(m=90, b=10)$ at age $x=65$ (left panel) and $x=75$ (right panel). At relatively high (historical) valuation rates there is little difference in prices. The refund has little impact. But at (current) low $r$, the difference can be substantial.}
\label{FIG1}
\end{center}
\end{figure}

\clearpage

\begin{figure}
\begin{center}
\includegraphics[width=0.95\textwidth]{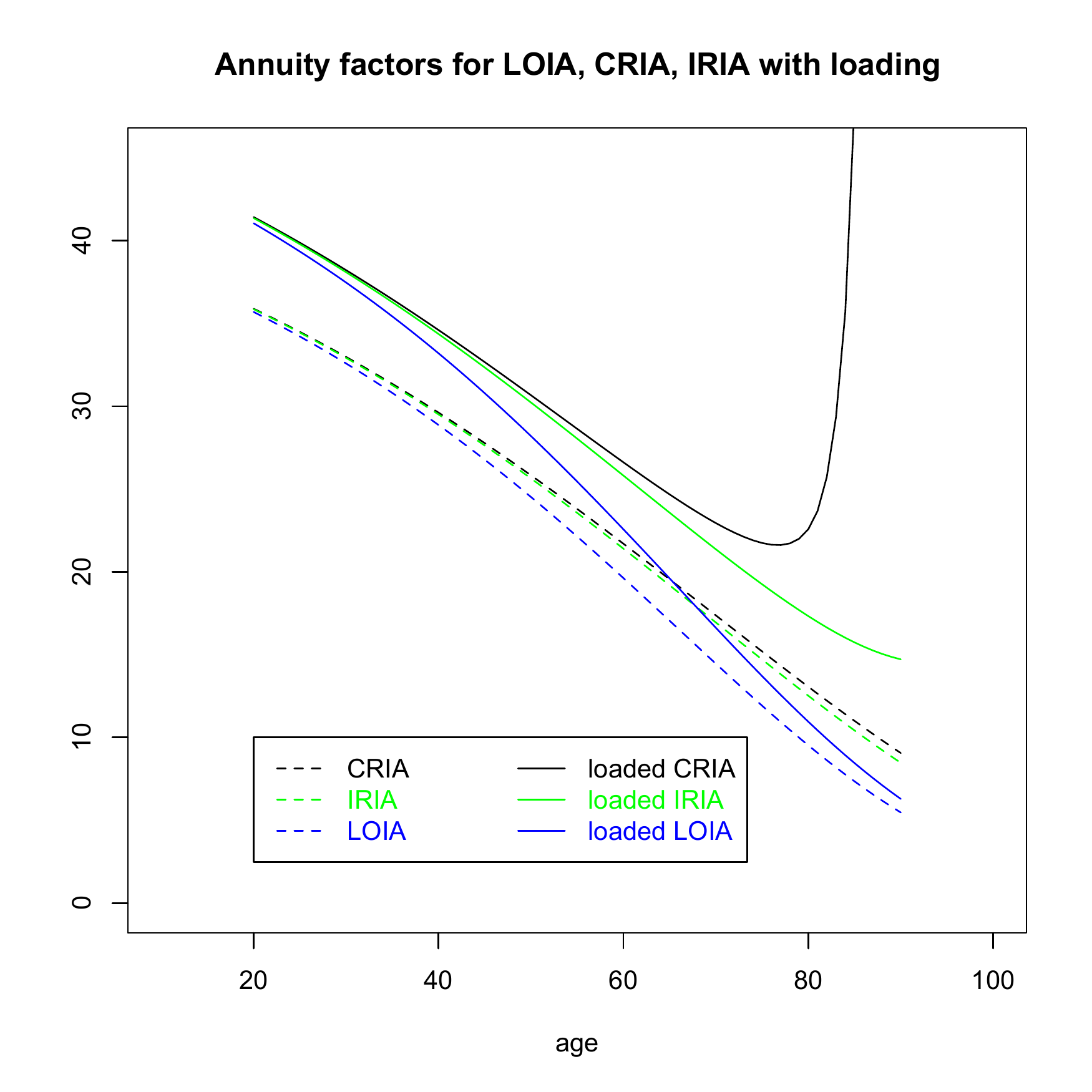} 
\caption{Income Annuity (IA) prices under a valuation rate: $r=2\%$, with Gompertz mortality $(m=90, b=10)$, comparing no loading with an insurance loading of $\pi=15\%$. Notice how the loaded CRIA price increases beyond (roughly) the age of $x=79$, and is no longer viable by (roughly) age $x=81$, in the above figure. }
\label{FIG2}
\end{center}
\end{figure}

\clearpage

\begin{figure}
\begin{center}
\includegraphics[width=0.45\textwidth]{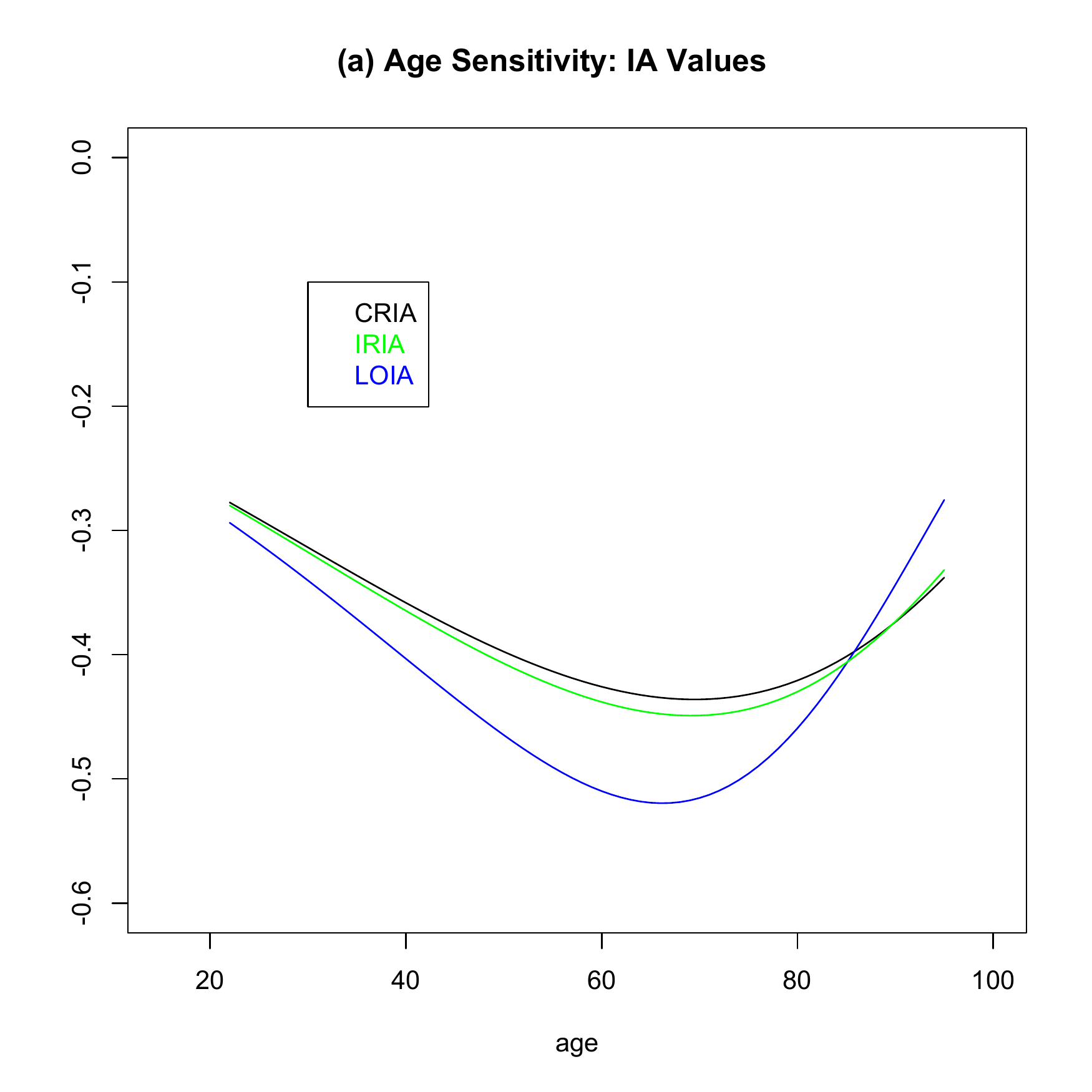} 
\includegraphics[width=0.45\textwidth]{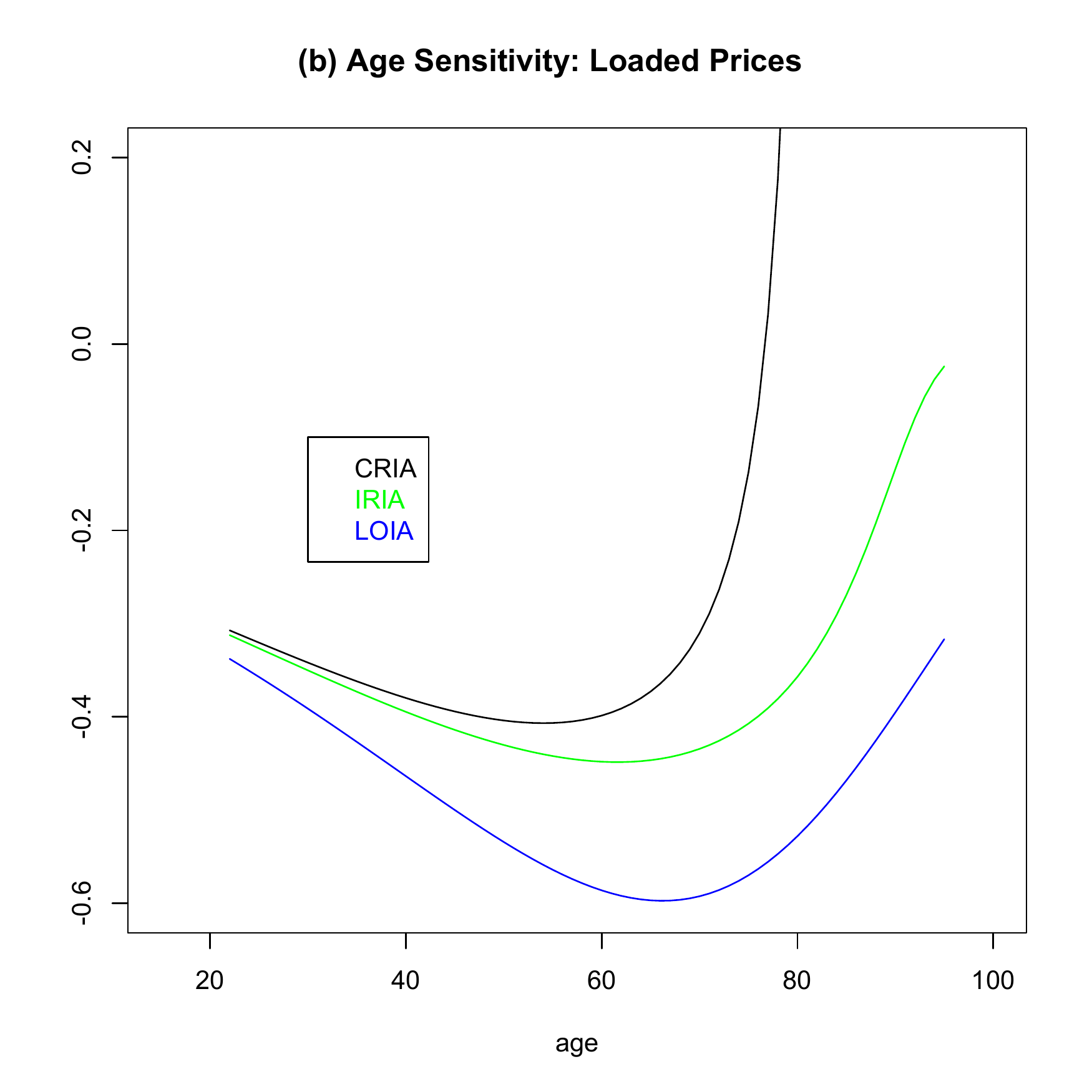} 
\caption{Sensitivity of Income Annuity prices to age $x$, technically $\frac{\partial}{\partial x}$, under a valuation rate of $r=2\%$ and Gompertz mortality $(m=90, b=10)$. The partial derivatives of unloaded IA prices are all negative since those pricess uniformly decline in age, but for the loaded $\pi=15\%$ CRIA prices, the numbers can be positive; aging increases the annuity price!}
\label{FIG3}
\end{center}
\end{figure}

\clearpage

\begin{figure}
\begin{center}
\includegraphics[width=0.45\textwidth]{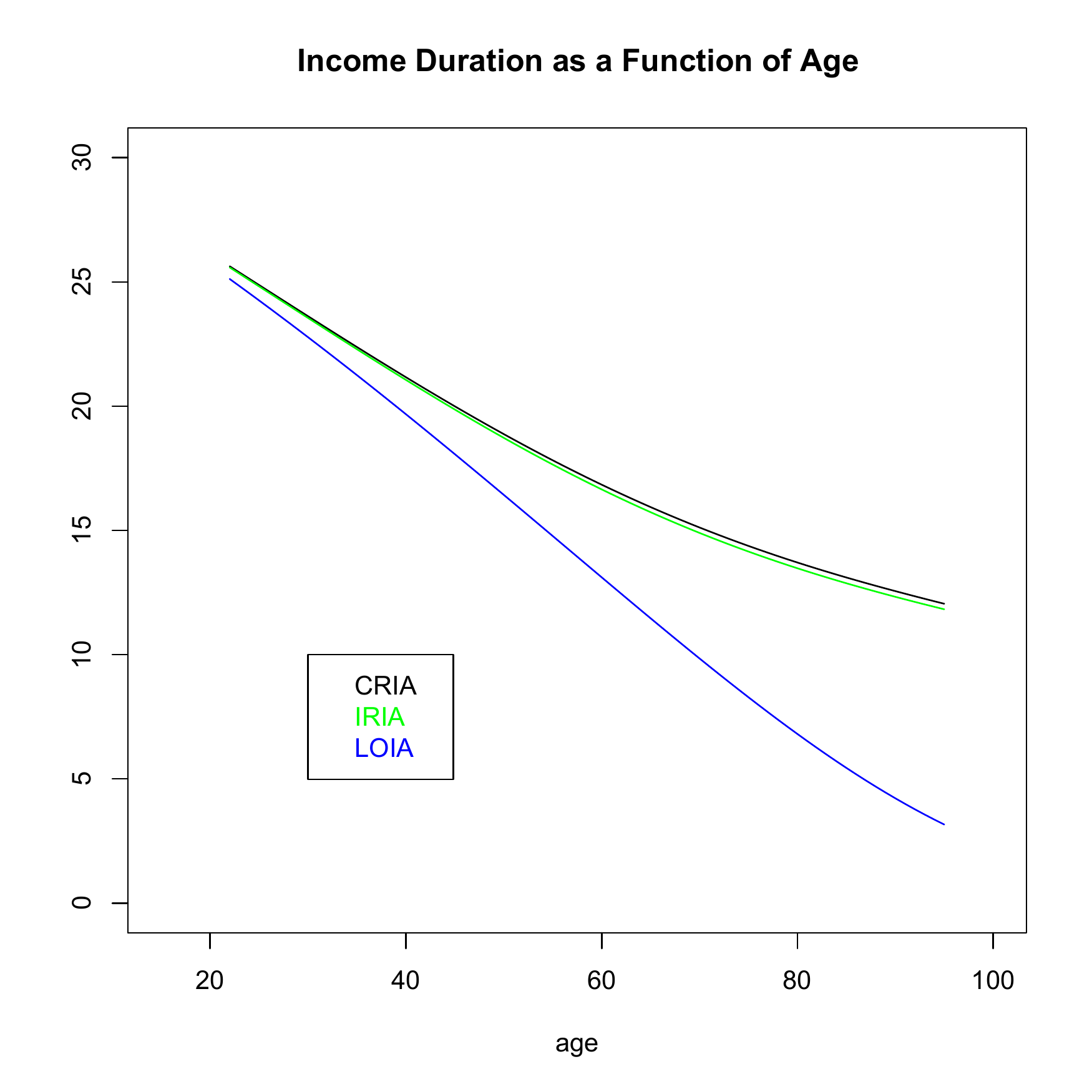} 
\includegraphics[width=0.45\textwidth]{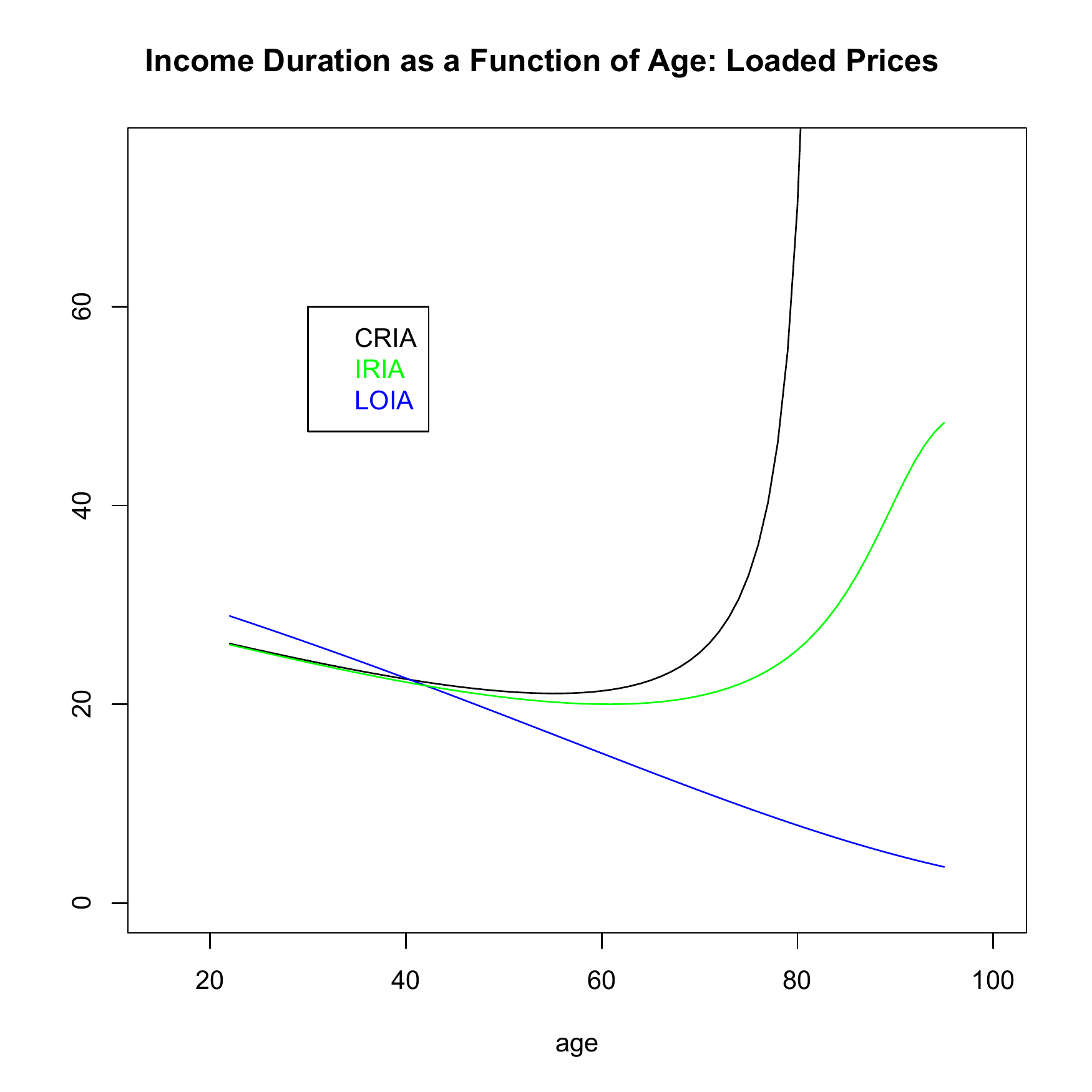} 
\caption{As defined and explained in the paper, {\em life annuity durations} at various issue ages, under a valuation rate of $r=2\%$, and Gompertz mortality $(m=90, b=10)$. The sensitivity of the unloaded actuarial price to interest rates is (much) higher for the cash-refund and instalment refund IA, especially at higher ages. For the loaded ($\pi=15\%$) IA, the situation is reversed at younger ages, and the life annuity duration of the cash-refund price ``blows up'' as $x$ reaches the critical age beyond 80.}
\label{FIG4}
\end{center}
\end{figure}

\end{document}